\documentclass[11pt]{article}

\usepackage{booktabs}   
\usepackage{subcaption} 

\usepackage[algo2e, vlined]{algorithm2e}
\usepackage{multicol}
\usepackage{amsthm}
\usepackage{graphicx}
\usepackage{url}
\usepackage[margin=1.0in]{geometry}
\newtheorem{theorem}{Theorem}[section]

\usepackage{color}
\def\pcode{\texttt}
\def\code{\pcode}

\usepackage{placeins}
\usepackage{enumitem} 
\usepackage{hyperref}
\usepackage{cite}       
\SetKwProg{Fn}{Function}{}{}
\SetKwProg{Lmd}{Lambda}{}{}
\SetKwProg{Class}{Class}{}{}
\SetKw{And}{\textbf{and}}
\SetKw{Or}{\textbf{or}}
\SetKw{To}{\textbf{to}}

\SetCommentSty{mycommfont}
\LinesNumbered

\makeatletter
\newcommand{\removelatexerror}{\let\@latex@error\@gobble}
\let\old@printtopmatter\@printtopmatter
\makeatother

\begin{document}

%
\title{Understanding and Optimizing Persistent Memory Allocation\thanks{%
This work was supported in part by NSF grants
CCF-1422649,     
CCF-1717712,    
and CNS-1900803,
and by a Google Faculty Research award.
}}

%
\author{Wentao Cai, Haosen Wen, H. Alan Beadle, Chris Kjellqvist, \\Mohammad Hedayati, and Michael L. Scott\\[2ex]
	Technical Report \#1008\\[1ex]
	Department of Computer Science, University of Rochester\\
	\{wcai6,hwen5,hbeadle,ckjellqv,hedayati,scott\}@cs.rochester.edu}
\date{March 2020}

%

\maketitle

\begin{abstract}
The proliferation of fast, dense, byte-addressable nonvolatile memory
suggests that data might be kept in pointer-rich ``in-memory''
format across program runs and even process and system crashes.  For
full generality, such data requires dynamic memory allocation, and while
the allocator could in principle be ``rolled into'' each data structure,
it is desirable to make it a separate abstraction.

 
Toward this end, we introduce \emph{recoverability}, a correctness
criterion for persistent allocators, together with a nonblocking
allocator, \emph{Ralloc}, that satisfies this criterion.  Ralloc is
based on the \emph{LRMalloc} of Leite and Rocha, 
with three key innovations.  First, we persist just enough
information during normal operation to permit correct reconstruction of
the heap after a full-system crash.  Our reconstruction mechanism
performs garbage collection (GC) to identify and remedy any
failure-induced memory leaks.  Second, 
we introduce
the notion of \emph{filter functions}, which identify the locations of
pointers within persistent blocks to mitigate the limitations of
conservative GC.
Third, to allow
persistent regions to be mapped at an arbitrary address, we employ
position-independent (offset-based) pointers for both data and metadata. 

Experiments show Ralloc to be performance-competitive
with both \emph{Makalu}, 
the state-of-the-art lock-based persistent allocator, and
such transient allocators as LRMalloc and JEMalloc. 
In particular, reliance on GC and offline metadata reconstruction
allows Ralloc to pay almost nothing for persistence during normal
operation.
\end{abstract}

\section{Introduction}\label{sec:intro}


Byte-addressable nonvolatile memory (\emph{NVM}) offers significantly higher
capacity and lower energy consumption than DRAM, with a latency penalty
of less than an
order of magnitude. Intriguingly, NVM also raises the possibility that
applications might access persistent data directly with load and store
instructions, rather than serializing updates through a block-structured file
system. 
Taking advantage of persistence, however, is not a trivial exercise. If a
persistent data structure is to be recovered after a full-system crash, the
contents of NVM must always represent a consistent logical state---ideally one
that actually arose during recent pre-crash
execution~\cite{izraelevitz2016linearizability}. Ensuring such consistency
generally requires that code be instrumented with explicit write-back and fence
instructions, to avoid the possibility that updated values may still reside only
in the (transient) cache when data that depend upon them have already been
written back.
To save the programmer the burden of hand instrumentation, various
groups have developed 
persistent versions of popular data
structures~\cite{friedman2018persistent, xu2016nova, yang2019orion,
nawab2017dali} as well as more general techniques to add
\emph{failure atomicity}~\cite{chakrabarti2014atlas} to 
code based on locks~\cite{chakrabarti2014atlas, izraelevitz2016justdo,
  liu2018ido}, transactions~\cite{volos2011mnemosyne, beadle2019qstm,
  ramalhete2019onefile, coburn2011nvheaps, correia2018romulus}, or
both~\cite{rudoff2014pmdk}. 

Whether written by hand or with the aid of an automated system, any
dynamically resizable persistent structure requires a
memory allocator that also tolerates power failures.
One could in principle insist that allocator operations
be integrated into the failure-atomic operations performed
on persistent structures, but this has the disadvantage of introducing
dependences among otherwise independent structures that share
the same allocator.  It also imposes a level of consistency (typically
\emph{durable linearizability}~\cite{izraelevitz2016linearizability})
that is arguably unnecessary for the allocator: we do not in general care
whether calls to \code{malloc} and \code{free} linearize so long as no
block is ever used for two purposes simultaneously or is permanently
leaked.

As in work on transactional memory~\cite{hudson-ismm-2006},
it is desirable to provide \code{malloc} and \code{free} as
primitives to the authors of persistent data structures.
In so doing, one must consider how to avoid memory leaks if a crash
occurs between allocating a block and attaching it persistently to the
data structure---or between detaching and deallocating it.
Intel's Persistent Memory Development Kit (PMDK)~\cite{rudoff2014pmdk}
exemplifies one possible approach, in which
the allocator provides a \code{malloc-to} operation that allocates a
block and, atomically, attaches it persistently at a specified address.
A similar \code{free-from} operation breaks the last persistent pointer
to a block and, atomically, returns it to the free list.  HPE's
Makalu~\cite{bhandari2016makalu} exemplifies an alternative approach,
in which a traditional \code{malloc}/\code{free} interface is supplemented
with post-crash garbage collection to recover any blocks that might
have leaked.

We adopt Makalu's approach in our work.  In addition to making it easier
to port existing application code, the traditional interface allows us
to eliminate write-back and fence instructions in allocator code, and
frees the programmer of the need to keep track (in persistent memory) of
nodes that have been allocated but not yet attached to the main data
structure---perhaps because of speculation, or because they are still
being initialized.

As a correctness criterion for a persistent allocator, we introduce the
notion of \emph{recoverability}.  Informally, we say an allocator is
recoverable if, in the wake of post-crash recovery, it ensures that the
metadata of the allocator will indicate that all and only the ``in use''
blocks are allocated.  In a \code{malloc-to}/\code{free-from} allocator,
``in use'' blocks would be defined to be those that have (over the
history of the structure, including crashes) been \code{malloc-to}-ed
and not subsequently \code{free-from}-ed.  In a
\code{malloc}/\code{free} allocator with GC, ``in use'' blocks are
defined to be those that are \emph{reachable} from a specified set of
\emph{persistent roots}.  
In this case, the application and the
allocator must agree on a tracing mechanism that enumerates the reachable blocks.
In the case of conservative tracing it is conceivable that a block will
appear to be reachable without ever having been allocated; by treating
such blocks as ``in use,'' we admit the possibility that a crash will
leak some memory.  As in prior work~\cite{boehm1988gc}, this never
compromises safety, and leaked blocks may often be recovered in
subsequent collections, if values erroneously interpreted as references
have changed.

Given a notion of correctness, we present what we believe to be the
first nonblocking recoverable allocator.  Our system, \emph{Ralloc}, is
based on the (transient) \emph{LRMalloc} of Leite and
Rocha~\cite{leite2018lrmalloc}, which is in turn derived from Michael's
nonblocking allocator~\cite{michael2004scalable}.
Like LRMalloc, Ralloc uses thread-local caching to fulfill most allocation
and deallocation requests without synchronization.  When it does need to
synchronize, it commonly issues two \emph{compare-and-swap} (CAS)
instructions and a write-back \& fence pair in \code{malloc} or
\code{free}. Most metadata needed for fast operation resides only in
transient memory (with no explicit writes-back required) and is
reconstructed after a full-system crash.  In the event of a partial
crash (e.g., due to a software bug outside the allocator that takes down
one of several cooperating processes), memory may be leaked on a
temporary basis: it can be recovered via garbage collection in some
subsequent quiescent interval.

For type-unsafe languages like C and C++, Ralloc adopts the conservative
strategy of Boehm and Weiser~\cite{boehm1988gc}.  To accelerate recovery,
accommodate nonstandard pointer representations, and reduce the likelihood of
erroneously unrecoverable blocks due to false positives during tracing, we
introduce what we call \emph{filter functions}---optional, user-provided
routines to enumerate the pointers in a given block.

To allow persistent structures to be mapped at different virtual addresses in
different processes and at different times, we use an offset-based
pointer representation~\cite{chen2017offholder, coburn2011nvheaps}
to provide fully \emph{position-independent data}. (Specifically, each
pointer stores the 64-bit offset of the target from the
pointer itself.)  By contrast, several existing systems force data to reside at
the same address in all processes across all of time~\cite{volos2011mnemosyne,
bhandari2016makalu}; others expand the size of each pointer to 128 bits for
base-plus-offset addressing~\cite{oukid2017pallocator, rudoff2014pmdk}.  The
former approach introduces an intractable bin-packing problem as application
needs evolve, and is incompatible with \emph{address space layout randomization}
(ASLR)~\cite{shacham2004random} for security; the latter introduces space
overhead and forces the use of a \emph{wide-compare-and-swap} (WCAS) for atomic
updates.

Summarizing contributions:
\begin{itemize}[leftmargin=1em,parsep=.5ex plus .5ex]
\item
    We introduce \emph{recoverability} as the correctness criterion for
    persistent allocators, eschewing unnecessary ordering among
    allocator operations and preserving the essential properties of
    conventional transient allocators for a world with persistent
    memory.
\item
    We introduce \emph{Ralloc}, the first nonblocking persistent allocator,
    providing allocation and deallocation operations that are fast and
    recoverable, that provide a standard API, and that can accommodate
    full-system failures and potentially independent thread failures. 
\item
    We introduce the notion of \emph{filter functions}, which allow the
    programmer to provide precise type information to enhance the performance,
    generality, and accuracy of conservative garbage collection.
\item
    We employ offset-based smart pointers throughout our code,
    providing persistent, in-memory data structures with fast
    \emph{position-independence}.
\item
    We present performance results confirming that Ralloc scales well to
    large numbers of threads and is performance competitive, on both
    allocation benchmarks and real applications, with both
    JEMalloc \cite{evans2006jemalloc}, a high-performance
    transient allocator, and Makalu~\cite{bhandari2016makalu},
    the state-of-the-art lock-based persistent allocator.
\end{itemize}


\section{System Model}\label{sec:model}





\subsection{Hardware and Operating System}
\label{sec:HW_and_OS}

We assume a hardware model in which NVM is attached to the system in
parallel with DRAM, and directly exposed to the operating system (OS)
as byte-addressable memory.  This model corresponds (but is not limited)
to recent Intel machines equipped with Optane DIMMs configured in
so-called \emph{App Direct}
mode~\cite{izraelevitz2019performance}.\footnote{%
    Intel also supports an alternative configuration, \emph{memory mode},
    in which DRAM serves as a hardware-managed cache for the larger but
    slower Optane memory, whose persistence is ignored.  We ignore this
    alternative in our work.}
The OS, for its part, makes NVM available to applications through a
\emph{direct access} (DAX)~\cite{rudoff2017persistent}
mechanism in which persistent memory segments have file system names and
can be mapped directly into user address spaces.
DAX employs a variant of \code{mmap} that bypasses the traditional
buffer cache~\cite{chinner2015xfs, wilcox2017ext4, xu2016nova,
  yang2019orion, williams2019persist}.  Like traditional files and
memory segments, DAX consumes physical NVM on demand: a page of physical
memory is consumed only when it is accessed at the first time.  This
feature allows the programmer to define memory segments that are large
enough to accommodate future growth, without worrying about space lost
to internal fragmentation.

Some DAX operations (e.g., \code{mmap}) have effects that are entirely
transient: they are undone implicitly on a system shutdown or crash.
We assume that all others are failure atomic---that is, the OS has been
designed (via logging and boot-time recovery) to
ensure that they appear, after recovery, to have happened
in their entirety or not at all.
By contrast, updates to DAX files mapped into user-level programs are
ordinary memory loads and stores, filtered through volatile caches that
reorder writes-back during normal operation, and that lose their
contents on a full-system crash.

Applications that wish to ensure the consistency of persistent memory
after a crash must generally employ special hardware instructions to
control the order in which cache lines are written back to NVM\@.
On recent Intel processors~\cite{rudoff2017persistent, intel-manual},
the \code{clflush} instruction evicts a line from all caches in the
system, writes it back to memory if dirty, and performs a store fence to
ensure that no subsequent store can occur before the write-back.  The
\code{clflushopt} instruction does the same but without the store fence;
\code{clwb} performs the write-back without necessarily evicting or
fencing.  The latter two instructions can be fenced explicitly with a
subsequent \code{sfence}.  In keeping with standard (if somewhat
inaccurate) usage in the literature, the rest of this paper uses
``flush'' to indicate what will usually be a
\code{clwb} instruction, and uses ``fence'' for \code{sfence}.

We assume that a persistent data structure must, at the very least,
tolerate \emph{full system, fail-stop} crashes, as might be caused by
power loss or the failure of a critical hardware component.  On such a
crash, dirty data still in cache may be lost, but writes-back at
cache-line granularity will never be torn, and there is no notion of
Byzantine behavior.

More ambitiously, we wish to accommodate systems what share data among
mutually untrusting applications with independent failure modes.
This stands in contrast to previous projects, which have assumed that
all threads sharing a given persistent segment are part of a single
application, and are equally trusted.
Recent work~\cite{hedayati2019hodor, erim} has shown that it is
possible, at reasonable cost, to amplify access rights when calling into
a \emph{protected library} and to reduce those rights on return.
The OS, moreover, can arrange for any thread currently executing in a
protected library to finish the current operation cleanly in the event
its process dies (assuming, of course, that the library itself does not
contain an error).  Applications that trust the library can then share
data safely, without worrying that, say, a memory safety error in
another application will compromise the data's integrity.
Protected libraries have the potential to greatly increase performance,
by allowing a thread to perform an operation on shared memory directly,
rather than using interprocess communication to ask a server to perform
the operation on its behalf.
They introduce the need to accommodate situations in which recovery from
process crashes, if any, proceeds in parallel with continued execution
in other processes.

We assume that the OS allows a manager process to be
associated with a protected library, and that it notifies this manager
whenever a process sharing the library has crashed (but the system as a
whole has not).
For whole-system crashes, we include a ``clean
shutdown'' flag in each persistent segment.  If a process (either an
individual user of a persistent structure or the manager of a shared
library) discovers on startup that this flag is not set, it can initiate
segment-specific recovery before continuing normal execution.

\subsection{Runtime and Applications}

We assume that every persistent data structure (or group of related
structures) resides in a \emph{persistent segment} that has a name in the
DAX file system and can be mapped into contiguous virtual addresses in
any program that wants to use it (and that has appropriate file system
rights).  The goal of our allocator is to manage dynamically allocated
space within such segments.  We assume, when the structure is quiescent
(no operations active), that any useful block will be reachable from some
static set of \emph{persistent roots}, found at the beginning of the
segment.  We further assume that an application will be able to tell
when it is the first active user of any given segment, allowing it to
perform any needed recovery from a full-system crash (if the segment was
not cleanly closed) and to start any additional processes needed for
segment management.  (Such processes might be responsible for background
computation or for segment-specific recovery from individual application
crashes.)

For all persistent data, we assume that application code takes
responsibility for \emph{durable
  linearizability}~\cite{izraelevitz2016linearizability,
  friedman2018persistent} or its buffered variant.  Durable
linearizability requires that data structure operations persist, in
linearization order, before returning to their callers.  Buffered
durable linearizability relaxes this requirement to allow some completed
operations to be lost on a crash, so long as the overall state of the
system (after any post-crash recovery) reflects a consistent cut across
the happens-before order of data structure operations.  Both variants
extend in a straightforward fashion to accommodate fail-stop crashes of
a nontrivial subset of the active threads.  They do not encompass cases
in which a crashed thread recovers and attempts to continue execution
where it last left off.  New threads, however, may join the execution.

While some data structures may be designed specifically for persistence
and placed in libraries, the requisite level of hand instrumentation is
beyond most programmers.  To facilitate more general use of persistence, 
several groups have developed libraries and, in
some cases, compiler-based systems to provide failure atomicity for
programmer-delimited blocks of code.  In some systems, these blocks are
outermost lock-based critical sections, otherwise known as
\emph{failure-atomic sections} (FASEs); examples in this camp include
Atlas~\cite{chakrabarti2014atlas}, JUSTDO~\cite{izraelevitz2016justdo},
and iDO~\cite{liu2018ido}.  In other systems, the atomic blocks are
\emph{transactions}, which may be speculatively executed in parallel
with one another; examples in this camp include
Mnemosyne~\cite{volos2011mnemosyne}, NV-Heaps~\cite{coburn2011nvheaps},
QSTM~\cite{beadle2019qstm}, and OneFile~\cite{ramalhete2019onefile}.
Intel's PMDK library~\cite{rudoff2014pmdk} also provides a transactional
interface, but solely for failure atomicity, not for synchronization
among concurrently active threads.

\section{Recoverability}
\label{sec:recoverability}

Our allocator needs to be compatible with all the programming models
described in the previous section.
As described in Section~\ref{sec:intro}, it must also address the
possibility that a crash may occur after a block has been allocated but
before it has been made reachable from any persistent root---or after it
has been detached from its root but before it has been reclaimed.
Rather than force these combinations to persist atomically, together, we
rely on post-crash garbage collection to recover any memory that leaks.
While GC-based systems require a mechanism to trace the set of in-use
blocks, they have compelling advantages. Use of a standard
API avoids the need to specify attachment points in calls to \code{malloc-to}
and \code{free-from}; it also facilitates porting of existing code.  More
importantly, the garbage collector's ability to reconstruct the state of the
heap after a crash avoids the ongoing cost of flushing and fencing both
allocator metadata and, for nonblocking structures, the \emph{limbo
  lists} used for safe memory reclamation~\cite{michael2004hazard,
  wen2018ibr, fraser-thesis-2004}.

%

We say that a persistent allocator is \emph{recoverable} if, in the wake
of a crash, it is able to bring its metadata to a state in which all and
only the ``in use'' blocks are allocated.
For applications using the \code{malloc-to}/\code{free-from} API,
``in use'' blocks can be defined to be those that have (over the history
of the structure, including crashes) been \code{malloc-to}-ed and not
subsequently \code{free-from}-ed. In a \code{malloc}/\code{free}
allocator with GC, we define ``in use'' blocks to be those that are
reachable from the persistent roots.
The notion of reachability, in turn, requires a mechanism to identify
the pointers in each node of the data structure, so that nodes can be
traced recursively.  If the identification mechanism is conservative,
then some blocks that were never actually allocated prior to a crash may
be considered to be ``in use'' after recovery.

We observe that, given an appropriate tracing mechanism, almost any correct,
transient memory allocator can be made recoverable under a full-system-crash
failure model.  During normal operation, no block will be leaked or used
for more than one purpose simultaneously; in the wake of a crash, a
fresh copy of the allocator can be reinitialized to reflect the
enumerated set of in-use blocks.  (In a type-safe language, the
reinitialization process may also perform compaction.  This is not
possible with conservative collection, since we do not always know
whether to update a word that appears to point at a relocated block.)
Very little in the way of allocator metadata needs to be saved
consistently to NVM\@.
The observation transforms the central question of persistent allocation
from ``how do we persist our \code{malloc} and \code{free} operations?''
to ``how do we trace our data structures during recovery?''

But not all allocators are created equal.  There are compelling reasons,
we believe, why a persistent allocator should employ \emph{nonblocking}
techniques.  First, a blocking allocator inherently compromises the
progress of any otherwise nonblocking data structure that relies on it
for memory management.  Second, nonblocking algorithms dramatically
simplify the task of post-crash recovery, since execution can continue
from any reachable state of the structure (and the allocator).  Third,
even if cross-application sharing employs a protected library that
arranges to complete all in-flight operations in a dying process, the
problem of priority inversion suggests that a thread should never have
to wait for progress in a different protection domain.


Among existing transient allocators, the first nonblocking
implementation is due to Maged Michael~\cite{michael2004scalable}.
It makes heavy use of the CAS instruction in allocation and deallocation
and is noticeably slower than the fastest lock-based allocators.
The more recent LRMalloc of Leite and Rocha~\cite{leite2018lrmalloc}
uses thread caching to reduce the use of CAS on its ``fast path,'' and
makes allocations and deallocations mostly synchronization free.  Other
lock-free allocators include NBMalloc~\cite{gidenstam2010nbmaloc}
and SFMalloc~\cite{sangmin2011sfmalloc}.
Due to the complexity of their internal data structures, these appear
much harder to adapt to persistence; our own work is based on LRMalloc.

\label{sec:conservative-gc}


Despite the development of nonblocking allocators, fast, nonblocking
\emph{concurrent} (online) collection remains an open research problem.
We adopt the simplifying assumption that crashes are rare and that a
blocking approach to GC will be acceptable in practice.  It is clearly so
in the wake of a full-system crash, when there are no application
threads to block.  In the event of partial (single-process) failures,
memory may leak temporarily.  If the allocator identifies a low-memory
situation and knows that one or more processes have crashed since GC was
last performed, it can initiate a stop-the-world collection.

\section{Ralloc}\label{sec:imp}

Ralloc is based on the lock-free LRMalloc
allocator~\cite{leite2018lrmalloc}.  From that system it inherits the
notion of thread-local caches of free blocks; allocations and
deallocations move blocks from and to these caches in most cases,
avoiding synchronization.  Operations that synchronize are rare; when
they occur, they typically incur two CAS instructions.

In adapting LRMalloc to persistence, we introduce four principal
innovations:
\begin{enumerate}[leftmargin=1.5em,parsep=.5ex plus .5ex]
\item
We rely on the fact that all blocks in a given \emph{superblock} (major segment
of the heap) are of identical size to avoid the need to persistently maintain a
size field in blocks.
Instead, we persist the common size during superblock allocation, which
is rare.
In the typical \code{malloc} operation, nothing needs to be written back
to memory explicitly.
\item
To improve the performance, accuracy, and generality of conservative
garbage collection, we introduce the notion of \emph{filter functions}.
These serve to enumerate the references found in a given block, for use
during trace-based collection.  In the absence of a user-provided filter
function, we fall back on traditional conservative collection, and
assume that every 64-bit aligned bit pattern is a potential reference.
\item
We reorganize LRMalloc's data into three respectively contiguous regions, in
which the major (superblock) region is utilized in increasing order of virtual
address as space demand increases. Corresponding physical pages will be
allocated by the OS on demand.
\item
To allow a persistent heap to be mapped at different addresses in different
applications (or instances of the same application over time), we use an
offset-based pointer representation~\cite{chen2017offholder,coburn2011nvheaps}
for Ralloc's metadata references, and encourage
applications to do the same for data structure pointers. The result can
aptly be described as \emph{position-independent data}. In our code
base, offsets are implemented as C++ smart pointers.
\end{enumerate}

\smallskip


\subsection{API}\label{sec:api}

\begin{figure} \removelatexerror
\begin{algorithm2e}[H]
	\Class{Ralloc}{
		\Fn{init$\,($string path, int size$)$ $:$ bool}{
			\tcp{create or remap heap in \emph{path}}
			\tcp{\emph{size} of superblock region}
			\tcp{return true if heap files exist}
		}
		\Fn{recover$\,()$ $:$ bool}{
			\tcp{issue offline GC and reconstruction if dirty}
			\tcp{return true if GC occurs, otherwise false}
		}
		\Fn{close$\,()$ $:$ void}{
			\tcp{give back cached blocks and flush heap}
			\tcp{set heap as not dirty}
		}
		\Fn{malloc$\,($int size$)$ $:$ void*}{
			\tcp{allocate \emph{size} bytes}
			\tcp{return address of allocated block}
		}
		\Fn{free$\,($void* ptr$)$ $:$ void}{
			\tcp{deallocate \emph{ptr}}
		}
		\Fn{setRoot$\,($void* ptr, int i$)$ $:$ void}{
			\tcp{set \emph{ptr} to be root \emph{i}}
		}
		\Fn{getRoot$\,<$T\,$>($int i$)$ $:$ T*}{
			\tcp{update root type info}
			\tcp{return address of root \emph{i}}
		}
	}
\end{algorithm2e}
\vspace*{-.5\baselineskip} 
\caption{API for Ralloc.}
\label{fig:api}
\vspace*{-\baselineskip} 
\end{figure}

The API of Ralloc is shown in Figure~\ref{fig:api}. Function \code{init()} must
be called to initialize Ralloc prior to using it.  This function checks the
persistent heap referenced by the parameter \code{path} to determine whether
this is a fresh start, a clean restart without unaddressed failure, or a dirty
restart from a crash. A fresh start will create persistent heap on NVM, map it
in DAX mode, initialize the metadata, and return \code{false}. A clean restart
will remap persistent heap to the address space and also return
\code{false}. A dirty restart will remap persistent heap to the address
space and return \code{true}, indicating that recovery is needed. If the
application receives \code{true} from \code{init()}, it will need to call
\code{recover()} (after calling \code{getRoot<T>()}---see below) to invoke the
offline recovery routine and reconstruct metadata.

It is the programmer's responsibility to provide a sufficiently large
\code{size} to \code{init()}.  If space runs out, calls to \code{malloc()} will
fail (return \code{null}). Resizing currently requires an allocator restart and
an \code{init()} call with a larger size. 
As a practical matter,
resizing only changes
the first word of the superblock region and calls \code{mmap} with a larger size; no data
rearrangement is required.

At the end of application execution, \code{close()} should be called to
gracefully exit the allocator.
In the process, any blocks held in thread-local caches will be returned
to their superblocks,
and the persistent heap will be written back to NVM for fast restart.

The functions used for allocation and deallocation are similar to the
traditional \code{malloc()} and \code{free()}. As a big picture, allocation and
deallocation requests are segregated by their size, into corresponding
\emph{size classes}. Most requests are fulfilled by thread-local caches of
blocks of each size class, avoiding synchronization. If the cache is
empty during
allocation, Ralloc either fetches a partially used \emph{superblock} (a chunk of
blocks in the same size) from the \emph{partial list} of the size class, or
fetches a free superblock from the \emph{superblock free list}. Both lists are
accessible by all threads. All free blocks in a fetched superblock will be
pushed into the thread-local cache. If the cache is full during deallocation,
all cached blocks will be transferred to their superblocks which may already or
then appear in a partial list. Our detailed implementation will be discussed in
Section~\ref{sec:alloc}.

In support of garbage collection, Ralloc maintains a set of
\emph{persistent roots} for data structures contained in the heap.
These serve as the starting points for tracing.
The \code{setRoot()} and \code{getRoot<T>()} routines are used to store
and retrieve these roots, respectively.  In C++, the \code{getRoot<T>()}
routine is generic in the type \code{T} of the root;
as a side effect, it associates with the root a pointer to the \code{T}
specialization (if any) of the GC filter function
(Section~\ref{sec:filter}) so that we avoid the need for position-independent
function pointers.  When \code{init()} returns \code{true}, the application
should call \code{getRoot<T>()} for each useful persistent root before it calls
\code{recover()}.
The application may use some enum type to decide which data structure is 
registered in which root, for easily tracing them back in restart.
It is thread safe to concurrently call \code{setRoot()} or \code{getRoot<T>()} with
different \code{i}, but not with the same \code{i}. 

\subsection{Data Structures}\label{sec:ds}

A Ralloc heap comprises a \emph{superblock region}, a \emph{descriptor
  region}, and a \emph{metadata region}, all of which lie in NVM
(Figure~\ref{fig:descsb}).
Only the fields shown in bold are flushed and fenced explicitly online.
(All fields are eventually written back implicitly, of course, allowing
quick restart after a clean shutdown.)  The three regions are respectively mapped
into the address space of the application.
The superblock region, which is by far the largest of the three,
begins with an indication of its maximum size,
which is set at initialization time and never changed.
A second word indicates the size of the prefix that is currently in use.
The descriptor region is always allocated at its maximum size, which can
be inferred from the \code{size} of the superblock region.
The metadata region has a fixed size, dependent on the number of size
classes, but not on the size of the heap.

\begin{figure}
\begin{center}
    \includegraphics[scale=0.65]{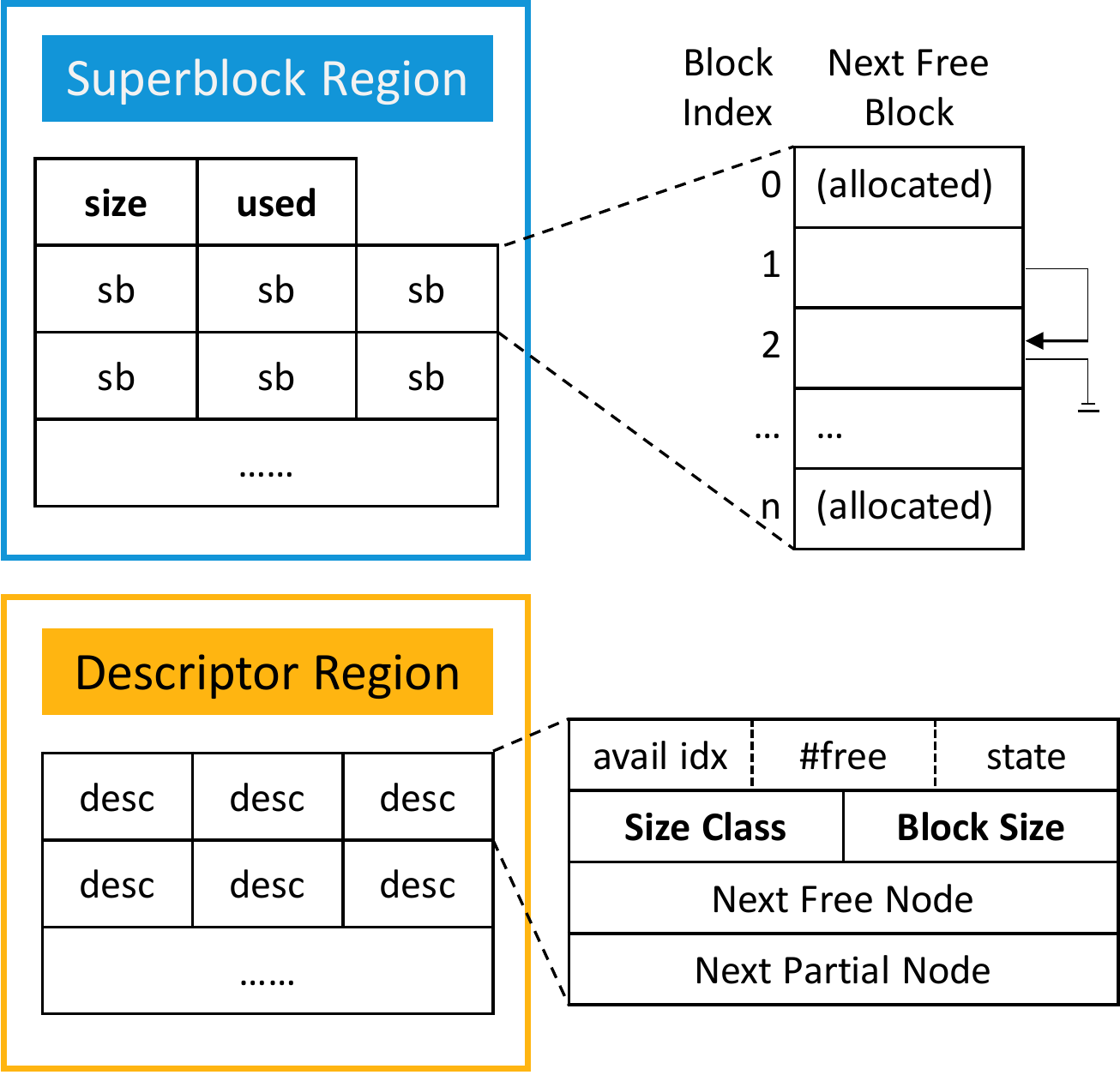}

    \vspace{0.25cm}

    \includegraphics[scale=0.65]{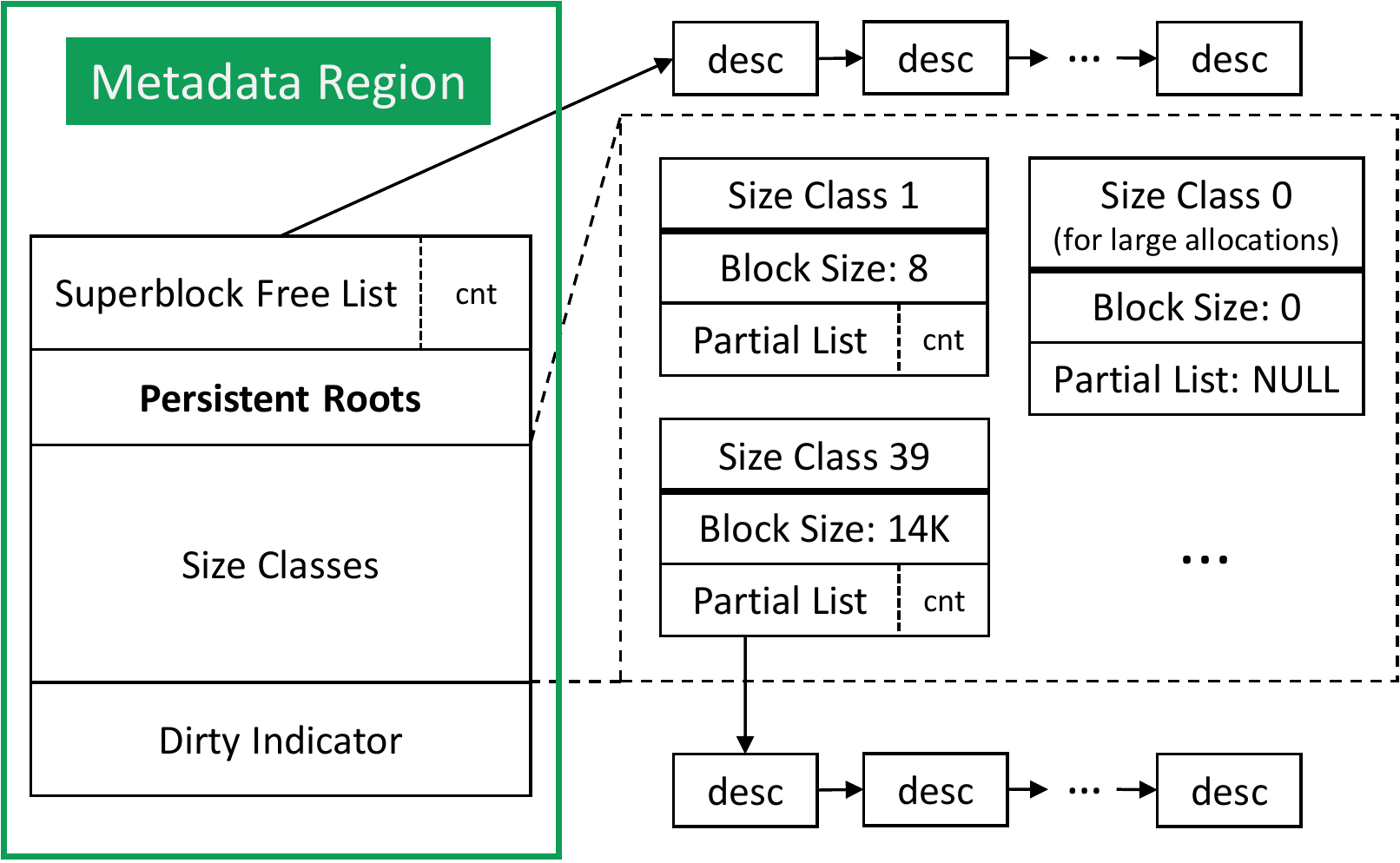}
\end{center}

\caption{Superblock, Descriptor, and Metadata Regions.
    Variables written back explicitly online are \textbf{bold}.}
\label{fig:descsb}
\vspace*{-\baselineskip} 

\end{figure}

The superblock region holds the actual data of the heap.
After the initial \emph{size} and \emph{used} words, it holds an array
of \emph{superblocks}. Each superblock is
64\,KB in size, and comprises an array of blocks. All blocks in a given
superblock are of the same size.  If a block is free (not in use), its
first word contains a pointer to the next free block, if any, in the
same superblock.

A \emph{descriptor} 
describes a superblock, and is the locus of synchronization on that
superblock.  Each descriptor
is 32\,B in size, padded out to a 64\,B cache line. Within a given heap, the
$i$-th descriptor corresponds to the $i$-th superblock, allowing either to be
found using simple bit manipulation given the location of the other.  Each
descriptor contains five fields: a 64\,b \emph{anchor}, a \emph{size
class} index, a \emph{block size}, and two optional pointers to form the
\emph{superblock free list} and a \emph{partial list}.  The anchor, which is
updated atomically with CAS, indicates the index of the first block on the block
free list, the number of free blocks, and the \emph{state} of the corresponding
superblock. The state is one of \code{EMPTY}, \code{PARTIAL}, or \code{FULL},
meaning the superblock is entirely free, partially allocated, or fully
allocated. The \emph{size class} field indicates which of several standard sizes
is being used for blocks in the superblock, or 0 if the superblock
comprises a single block that is larger than
any standard size. The \emph{block size} field indicates the size of each
block in this superblock, either fetched from a size class or the actual size of
the large block set during allocation. When the superblock of this descriptor is
in the superblock free list or a partial list, one of the descriptor's
pointer fields is set to the next node in the list.

The \emph{size class} field (and also \emph{block size} if it is a large
block) has to be persisted before a superblock is used for allocation
because Ralloc needs the size information of every reachable block to recover
metadata.

The metadata region holds the portion of Ralloc's metadata, other than
descriptors, that is needed on a clean restart.  Unlike LRMalloc, which
always calls \code{munmap()} to give empty superblocks back to the OS, Ralloc
implements a \emph{superblock free list}. This list is a lock-free LIFO
list (a Treiber stack~\cite{treiber-1986}) of descriptors, linked
through their \emph{next free node} fields. Given
the 1-to-1 correspondence between superblocks and descriptors, Ralloc finds a
free superblock easily given a pointer to its descriptor.

\emph{Persistent roots} point to the external entry points of persistent data
structures in the superblock region.  They comprise the starting points for
tracing during garbage collection: only blocks identified as potentially
reachable from the roots will be preserved; all other blocks will be identified
as unallocated. In our current implementation a metadata region contains 1024
roots; applications can initialize as many of these as required for a given set
of data structures.

Our current implementation supports 39 different \emph{size classes},
supporting blocks that range from 8 to 14\,K bytes~\cite{leite2018lrmalloc}.
A 40th class (number 0) supports blocks that are larger than those of
any standard class.
Each superblock holds blocks of exactly one class.
Each size class metadata record contains the block size and a
\emph{partial list};
elements on these LIFO lists are descriptors (linked through their \emph{next
partial node} fields) for partially filled superblocks whose blocks are of
the given size class.
The head of both partial lists and the superblock free list have 34 bits
devoted to a counter (a benefit of the persistent pointers discussed in
Section~\ref{sec:PIdata} below) to avoid the \emph{ABA
  problem}~\cite[Sec. 2.3.1]{SMS}.

Ralloc uses a \emph{dirty indicator}, implemented as a ``robust''
\code{pthread\_mutex\_t}~\cite{pthread-mutex}, to indicate whether it is
necessary to perform recovery before using the allocator.
Every time a process starts or restarts a Ralloc heap, Ralloc tries to lock the
mutex. If it fails with error code \code{EOWNERDEAD}, meaning that the
previously owning
thread terminated with the mutex locked, then the allocator was
not cleanly shut down.
During a normal exit, after all metadata has been written back to NVM,
the mutex is unlocked. 
The current implementation can only detect a full-system crash, since a
robust \code{pthread\_mutex\_t} records failure for only a single
application.
We will discuss, in Section~\ref{sec:independent}, a possible extension
of this mechanism to accommodate independent process failures.

In addition to its three persistent regions, Ralloc maintains transient
thread-local caches of blocks of each size class.  
In the event of a crash, all record of blocks held in thread-local caches will
be lost, and must be recovered via garbage collection.
On a clean shutdown, the thread-local caches are naturally empty.

Superblocks, descriptors, size classes, partial lists, and thread-local caches
are all inherited from LRMalloc.  Ralloc reorganizes 
them into three contiguous regions; adds persistent roots, the
superblock free list, and the dirty indicator; links descriptors rather
than superblocks in partial lists; and persists the necessary fields.

\subsection{Persistent Region Management}

In order to limit the length of the superblock free list, initially only
1\,GB of a heap's superblock region is included.  More
superblocks are made available on demand until the heap reaches the
\emph{size} limit specified in the most recent call to \code{init()}.
Within the specified limit, Ralloc obtains more space by CAS-ing the
\code{used} field to a greater number (with an explicit flush and
fence).
This update happens internally either when no superblock is available or
when a large allocation request is made.
The descriptor region always occupies its maximum size of
\emph{size}/1024 (a superblock is 64\,KB whereas a descriptor is 64\,B).
The metadata region has a fixed size, not proportional to the overall
size of the heap.


\subsection{Allocation and Deallocation}\label{sec:alloc}

Allocation requests for small objects are segregated by size class. The
thread-local cache of free small blocks is typically not empty, allowing the
request to be fulfilled without synchronization most of the time. An empty
thread-local cache will be refilled before satisfying the allocation request.
The cache is refilled with all available blocks in a partial
superblock, or with all blocks in a new superblock if the partial list is empty.
The anchor in the corresponding descriptor is updated with a CAS when the
superblock is used to refill the cache. A new superblock is taken from
the superblock free list.  If the list is empty, it is refilled by
expanding the used space of the superblock region.  Our current
implementation performs such expansion in 1\,GB increments.
We did not observe significant changes in performance with larger
or smaller expansion sizes.

When a small object is being deallocated, its descriptor is found via
bit manipulation. Ralloc determines its size class from the
descriptor. If the thread-local cache is not full, then the freed block
is simply added to the cache.  Otherwise, all of the blocks in the cache
are first pushed back to the block free list(s) of their respective
superblock(s). A descriptor changed from \code{FULL} to \code{PARTIAL}
is pushed to the partial list; a descriptor changed from \code{FULL} to
\code{EMPTY} is retired and pushed to the superblock free list.
A descriptor
that is changed from \code{PARTIAL} to \code{EMPTY}
will be retired, later, when it is fetched from the partial list.

Allocation and deallocation routines for small objects are inherited largely
from LRMalloc.  Given space constraints, the code is not shown here; it
differs from the original mainly in the addition of flush and fence
instructions needed to persist the fields shown in bold in Figure~\ref{fig:descsb}.
Large objects see a bit more change in the code.
In LRMalloc, any allocation over 14\,KB is fulfilled directly by \code{mmap} at
an arbitrary virtual address. This approach is not applicable in Ralloc because
all of our allocations must lie in the same persistent segment (i.e.,
the superblock region). Ralloc therefore 
rounds the size of a large allocation up to a multiple of the superblock size
(64\,KB) and allocates it by expanding the used space in the superblock region.
The size is persistently stored in the first corresponding descriptor. Although this
approach may introduce some external memory fragmentation, we consider it
acceptable if large allocations are rare. Ralloc with mostly small allocations
has no external fragmentation and little internal fragmentation.

When a large block is deallocated, it is split into its
constituent superblocks, which are then pushed to the superblock free list.
Both allocation and deallocation, for both small and large objects, are 
lock-free operations.
Updates to persistent fields (those shown in bold in Figure~\ref{fig:descsb})
are flushed and fenced to enable post-crash
recovery.  Other fields are reconstructed during recovery.

\subsection{Recovery}\label{sec:rec}

Recovery employs a tracing garbage collector to identify
all blocks that are reachable from the specified persistent roots.
Because the sizes of all blocks are determined by
their superblock (whose size field is persisted), it is easy to tell how
much memory is rendered reachable by any given pointer (pointers to
fields \emph{within} a block are not supported).
After GC, all metadata is reconstructed.
In a bit more detail, recovery comprises the following steps:
\begin{enumerate}[leftmargin=1.5em,parsep=0ex plus .5ex]
	\item Remap persistent regions to memory.\label{sp:1_s}
	\item Initialize thread-local caches as empty.
	\item Initialize empty superblock free and partial lists.\label{sp:1_e}
	\item Set the filter function for each persistent root. \label{sp:filter}
	\item Trace all blocks reachable from persistent roots and put their
	addresses in a transient set.\label{sp:gc}
	\item Scan superblock region and keep only traced blocks.\label{sp:sweep}
	\item Update each descriptor accordingly.
	\item Reconstruct the partial list in each size class.
	\item Reconstruct the superblock free list.\label{sp:recon}
	\item Flush the three persistent regions and issue a fence. \label{sp:gc_e}
\end{enumerate}

\smallskip

Steps \ref{sp:1_s}--\ref{sp:1_e} are done in \code{init()}. Step~\ref{sp:filter}
is done when \code{getRoot<T>()} is called for each persistent root. Steps
\ref{sp:gc}--\ref{sp:gc_e} are done in \code{recover()}.
When \code{init()} is called, the dirty indicator (see
Section~\ref{sec:ds}) is reinitialized and set dirty until a call to
\code{close()}; any crash that happens in the recovery steps leaves the 
allocator dirty.

\subsubsection{Filter Garbage Collection}\label{sec:filter}

Precise GC, of course, is impossible in C and C++, due to the absence of
type safety.  Conservative collection admits the possibility that some
64\,b value will erroneously appear to point to a garbage block,
resulting in a block that appears to be in use, despite the fact that it
was not allocated (or was freed) during pre-crash execution---in effect,
a memory leak.  Arguably worse, conservative collection is
incompatible with pointer tagging and other nonstandard representations.
\emph{Filter functions} serve to address these limitations.


\begin{figure}
\removelatexerror
\begin{algorithm2e}[H]

\Class{Ralloc}{
	\ldots\ \tcp{functions mentioned in API and metadata}
	\textsl{roots} : Persistent roots\,\;
	\emph{rootsFunc} : Functions for persistent roots\,\;
	\Fn{getRoot$\,<$T$\,>($int i$\,) :$ T*}{
		\emph{rootsFunc$[$i$\,]$} = \\
		\Lmd{$[\,]($void* p, GC\textsl{\&} gc$) :$ void}{\emph{gc.visit$\,<$T$\,>$}\,(\emph{p})\,\;}
		\Return{roots$[$i$\,];$}
	}
}

\Class{GC}{
	\emph{visitedBlk} : Set of visited blocks\,\;
	\emph{pendingBlk} : Stack of pending blocks to be visited\,\;
	\emph{pendingFunc} : Stack of functions of pending blocks\,\;

	\Fn{visit$\,<$T$\,>($T* ptr$) :$ void}{
		\If{ptr $\in$ superblock region \And ptr $\notin$ visitedBlk}{
			insert \emph{ptr} to \emph{visitedBlk}\,\;
			push \emph{ptr} to \emph{pendingBlk}\,\;
			push
				\Lmd{$[\,]($void* p, GC\textsl{\&} gc$) :$ void}{\emph{gc.filter$\,<$T$\,>$}\,(\emph{p})\,\;}
			to \emph{pendingFunc}\,\;
		}
	}

	\Fn{filter$\,<$T$\,>($T* ptr$) :$ void}{
		\tcp{Default conservative filter function}
		get descriptor \emph{desc} of \emph{ptr}\,\;
		read block size \emph{size} from \emph{desc} \tcp*[l]{0 if invalid}
		\For{i$\,\,=0$ \To size$\,-1$}{
			read potential pointer \emph{curr} at \emph{i}-th byte in \emph{ptr}\,\;
			\emph{visit}(\emph{curr})\,\;
		}
	}

	\Fn{collect$\,() :$ void} {
		\tcp{To get the set of reachable blocks}
		\For{i$\,\,=0$ \To max root}{
			\If{roots$[$i$\,]$ $\neq$ NULL}{
				\emph{rootsFunc$[$i$\,]$}(\emph{roots$[$i$\,]$},
                                \emph{*this})\,\;
			}
		}
		\While{pendingBlk $\neq \emptyset$}{
			pop \emph{func} from \emph{pendingFunc}\,\;
			pop \emph{blk} from \emph{pendingBlk}\,\;
			\emph{func}(\emph{blk}, \emph{*this})\,\;
		}
	}
}
\vspace*{-\baselineskip}
\end{algorithm2e}
\caption{Filtered garbage collection.}
\label{fig:gc}
\vspace*{-\baselineskip}
\end{figure}

Figure~\ref{fig:gc} shows the basic variables and functions related to filter GC.
The basic principle is that, when \code{getRoot<T>()} is called after 
\code{init()} but before \code{recover()}, its type (obtained via
template instantiation) is recorded in the transient array
\code{rootsFunc}, in the form of a lambda expression that calls the
\code{visit<T>()} function. 
Then in \code{recover()}, \code{collect()} traces all reachable
blocks by calling \code{visit<T>()} iteratively from persistent roots until no
more new blocks are found.
Each
\code{visit<T>()} function marks its block as reachable and then calls
\code{filter<T>()}, which is assumed to call 
\code{visit<U>()} for each pointer of type \code{U} in the block.

For each type \code{U} used for persistent data structures, the
programmer is 
encouraged to provide a corresponding \code{filter<U>()}. 
Figure~\ref{fig:filter} presents an example filter
for binary tree nodes.
If no \code{filter<U>()} has been specialized, the default conservative
filter, defined in Figure~\ref{fig:gc}, is called instead.

\begin{figure}
	\removelatexerror
	\begin{algorithm2e}[H]
	
	\Class{TreeNode}{
                \ldots\ \tcp{content fields}
		\emph{left}, \emph{right} : \emph{TreeNode}*\,\;
	}
	
	\Fn{filter$\,($TreeNode* ptr$\,) :$ void}{
		\emph{visit}(\emph{ptr}$\rightarrow$\emph{left});
                \emph{visit}(\emph{ptr}$\rightarrow$\emph{right})\,;\
	}
	\end{algorithm2e}
	\vspace*{-.2\baselineskip}
	\caption{Example of a filter function for binary tree nodes.}
	\label{fig:filter}
	\vspace*{-\baselineskip}
\end{figure}

While the implementation shown here utilizes C++ templates, filter functions
are easily adapted to pure C by arranging for \code{visit()} and
\code{getRoot()} to take a pointer to the appropriate filter function
as an extra parameter, and for filter functions themselves to pass the
appropriate function pointer in each of their calls to \code{visit()}.

Note that the function pointers used in GC are reestablished in each execution,
avoiding any complications due to recompilation or address space layout
randomization (ASLR)~\cite{shacham2004random}.  Mechanisms to tag
persistent roots with persistent type information are a potential topic
for future work.


\subsubsection{Sharing Across Processes}\label{sec:independent}

The mechanisms described above suffice to manage a persistent heap that
is used by
one application at a time.  While this application may be multithreaded,
its threads all live and die together.  If we wish to allow a heap to be
shared by threads in \emph{different} processes---either with mutual
trust or via protected libraries~\cite{hedayati2019hodor}---we must
address a pair of problems.  While neither is addressed in our current
implementation, solutions appear straightforward.

First, if a heap in a newly rebooted system may be mapped into more than
one process concurrently, we need a mechanism to determine which of
these processes is responsible for recovery.  While several strategies are
possible, perhaps the simplest assigns this task to a dedicated
\emph{manager process}.  Such a process could be launched
by any application that calls \code{init()} on a currently inactive
segment.

Second, we must consider the possibility that a process may crash (due
to a software bug or signal) while others continue to use the heap.
While nonblocking allocation ensures that the heap will remain usable
and consistent, blocks may leak for the same reasons as in a
full-system crash: they may be allocated but not yet attached, detached
but not yet deallocated, held in a per-thread cache, or held in a limbo
list awaiting safe reclamation.  Given our reliance on post-crash garbage
collection, these blocks can be recovered only by tracing from
persistent roots.  
As indicated at the end of Section~\ref{sec:conservative-gc}, we assume that
crashes are rare and that it will be acceptable to implement blocking,
``stop-the-world'' collection when they occur.  A likely implementation
would employ a failure detector provided by the operating system and
fielded by the manager process mentioned above.
In the wake of a single-process crash, the manager could initiate
stop-the-world collection using a quiescence mechanism adapted from
asymmetric locking~\cite{vasudevan-pact-2010}.



\subsection{Position Independence}\label{sec:PIdata}

There are several reasons not to implement pointers as absolute virtual
addresses in persistent memory.  If an application uses more than one
independent persistent data structure, the addresses of those structures
will need to be distinct.  If new applications can be designed to use
arbitrary existing structures, then every such structure would need
to have a globally unique address range, suggesting the need for global
management and interfering with security strategies like ASLR.

One option, employed 
by some
earlier work on InterWeave~\cite{chen-icpp-2002}, is to explicitly
relocate a heap when it is first mapped into memory, ``swizzling''
pointers as necessary.  Unfortunately, this approach requires precise type
information and still requires that all concurrent users map a heap at
the same virtual address.

Some systems (e.g., PMDK~\cite{rudoff2014pmdk}) use offsets from the
beginning of a destination segment rather than absolute addresses, but usually
consume 128 bits per pointer.
This \emph{based pointer} convention requires that the
starting address of the segment be available (e.g., in a reserved
register) in order to convert a persistent pointer to a virtual address.
An attractive alternative, used by NV-Heaps~\cite{coburn2011nvheaps} is
to calculate the offset not from the beginning of the segment but from
the location of the pointer itself, since that location is certain to be
conveniently available when storing to or loading from it.  Chen et
al.~\cite{chen2017offholder} call such offset-based pointers
\emph{off-holders}.


We implement off-holders as 64-bit \code{pptr<T>} smart pointers in C++, and
instruct programmers to use them instead \code{T*} pointers.  All of the
usual pointer operations work as one would expect, with no additional
source-code changes.  Chen et al.\ report
overheads for this technique of less than 10\%.

For cross-region metadata pointers within the same instance of Ralloc
(e.g., persistent roots that reside in the metadata region and point to
the superblock region), \code{pptr} takes an optional template parameter
as the index of a region. The default value indicates that this is an
off-holder pointing to a target in the current region.
Three other values can be used to indicate a based pointer for the
metadata, descriptor, or superblock region of the segment.  Ralloc
records the base address of regions during initialization, allowing it
to look up these addresses while converting a region-specific pointer to
an absolute address.  Note that such application programmers never need
based pointers: they occur only within the code of Ralloc itself.

Given a hard limit on the size of a superblock region (currently
1\,TB), Ralloc is able to repurpose some of the bits in a 64\,b
\code{pptr}. As noted in Section~\ref{sec:ds}, part of each list head
is used for an anti-ABA counter.  For an off-holder, the unused bits
hold an arbitrary uncommon pattern that is masked away during use; this
convention serves to reduce the likelihood that frequently-occurring
integer constants will be mistaken for off-holders during conservative
post-crash GC\@.

The \code{pptr} implementation does not currently support general
cross-heap references.  Among our near-term plans is to implement a
\emph{Region ID in Value (RIV)}~\cite{chen2017offholder} variant of
\code{pptr}, retaining the smart pointer interface and the size of 64 bits.


\section{Correctness}~\label{sec:proof} 

During crash-free execution, LRMalloc is both safe and live:
blocks that are concurrently in use (\code{malloc}ed and not yet
\code{free}d) are always disjoint (no conflicts), and blocks that are
\code{free}d are eventually available for reuse (no leaks).
We argue that Ralloc preserves these properties, and is additionally
lock free and recoverable.


\begin{theorem}[\textbf{Overlap freedom}]
    \label{theorem:lin}
    Ralloc does not overlap any in-use blocks.
\end{theorem}


\begin{proof}[Proof sketch]
This property is essentially inherited from LRMalloc.
All small allocations are eventually fulfilled from thread-local caches,
which are recharged by removing superblocks from the global free and
partial lists.  Large allocations, likewise, are fulfilled with entire
superblocks, obtained using CAS to update the \code{used} size field. Only the
CASing thread has the right to allocate from the new superblocks.

The global lists are lock-free Treiber stacks~\cite{treiber-1986}.
Only one thread at a time---the one that removes a superblock from a
global list---can allocate from the superblock.
Expansion of the heap (to create new superblocks) likewise happens in a
single thread, using CAS to update the \code{used} size field.

Small blocks in separate superblocks are disjoint, as are blocks within
a given superblock.  The blocks that tile a superblock change size only
when the superblock cycles through the free list; the superblocks that
comprise a large allocation likewise cycle through the free list.  Thus
blocks of different sizes that overlap in space never overlap in time.

A thread that cannot allocate from a given superblock may
still \emph{deallocate} blocks, but the free list within the superblock
again functions as a Treiber stack, with competing operations mediated
by CASes on the anchor of the corresponding descriptor.
These observations imply that allocations of the same block
never overlap in time.
\end{proof}

\begin{theorem}[\textbf{Leakage freedom}]
    Freed blocks in Ralloc eventually become available for reuse.
\end{theorem}


\begin{proof}[Proof sketch]
    Reasoning here is straightforward.  Small blocks, when deallocated,
    are returned to the thread-local cache.
    Superblocks are returned to global partial lists or (if
    \code{EMPTY}) to the global free list when the thread-local cache is
    too large, or when a large block is deallocated.  In either case,
    deallocated blocks are available for reuse (typically by the same
    thread; sometimes by any thread) as soon as \code{free} returns.
    Exactly \emph{when} reuse will occur, of course, depends on the
    pattern, across threads, of future calls to \code{malloc} and
    \code{free}.  Note that safe memory
    reclamation~\cite{michael2004hazard, wen2018ibr}, if any, is layered
    \emph{on top} of \code{free}: the Ralloc operation is invoked not at
    retirement, but at eventual reclamation.
\end{proof}

\begin{theorem}[\textbf{Liveness}]
    Ralloc is lock free during crash-free execution.
\end{theorem}

\begin{proof}[Proof sketch]
    The only unbounded loops in Ralloc occur in the Treiber-stack--like
    operations on the superblock free and partial lists and the free
    lists of individual superblocks, and in operations on anchors and the
    \code{used} size field.  In all cases, the failure of a CAS
    that triggers a repeat of a loop always indicates that another
    thread has made forward progress.

    Significantly, there are no explicit system calls (e.g., to
    \code{mmap}) inside Ralloc's allocation and deallocation routines.
    We assume that implicit OS operations, such as those related to
    demand paging and scheduling, always return within a reasonable
    time; we do not consider them as sources of blocking in Ralloc.
\end{proof}

\begin{theorem}[\textbf{Recoverability}]
\label{the:ralloc}
    Ralloc is recoverable.
\end{theorem}

Recall that an allocator is \emph{recoverable} if it ensures, in the
wake of post-crash recovery, that the metadata of the allocator will
indicate that all and only the ``in use'' blocks are allocated.  For
Ralloc, ``in use'' blocks are defined to be those that are reachable
from a specified set of persistent roots.  In support of this
definition, Ralloc assumes that the application follows certain
rules. Specifically:

\begin{enumerate}[leftmargin=1.5em,parsep=0ex plus .5ex]
    \item It is (buffered) durably linearizable (Ralloc does not transform a
      transient application to be persistent). \label{enum:per}
    \item It registers persistent roots in such a way that all blocks it
      will ever use in the future are reachable. \label{enum:reg}
    \item It eventually attaches every allocated block to the structure
      to make it reachable. \label{enum:link} 
    \item It eventually calls \code{free} for every detached block. \label{enum:detach} 
    \item It specializes a filter function for any block whose internal
      pointers are not 64-bit aligned \code{pptr}s. \label{enum:pptr}
\end{enumerate}

\begin{proof}[Proof sketch]

    These rules ensure that the application never leaks blocks during
    crash-free operation (rules~\ref{enum:link} and~\ref{enum:detach}), and
    that it enables GC-based recovery (rules~\ref{enum:reg}
    and~\ref{enum:pptr}).
    The size information of any in-use block in a descriptor is safely
    available to recovery via the \code{size} \code{class} and
    \code{block} \code{size} persistent fields.  Assuming that pointers
    in types without
    specialized filter functions are always aligned and in
    \code{pptr} format, Ralloc's garbage collection, with or without
    specialized filter functions, is guaranteed to find all in-use
    blocks by tracing from the persistent roots.  Having identified
    these blocks, Ralloc re-initializes its metadata accordingly,
    updating each descriptor and reconstructing free lists and partial
    lists.  By the end of the recovery, all and only the in-use blocks
    are allocated (where ``in use'' is defined to include all blocks
    found by the collector, even if they were not actually
    \code{malloc}ed during pre-crash execution).
\end{proof}

\section{Experiments}\label{sec:exp}
\subsection{Setup}
We ran all tests on 
Linux 5.3.7 (Fedora 30), on a machine
with two Intel Xeon Gold 6230 processors, each with 20 physical
cores and 40 hyperthreads. Threads were first pinned one per core
in the first socket, then on the extra hyperthreads, and finally on
the second socket. All experiments were conducted on 6 channels of
128\,GB Optane DIMMs, all in the first socket's NUMA
domain. Persistent allocators ran on an EXT4-DAX filesystem built on
NVM; transient allocators ran directly on raw
NVM~\cite{izraelevitz2019performance}.

We compared Ralloc to representative persistent and transient allocators
including \emph{Makalu}~\cite{hpe2017atlasmakalu}, \emph{libpmemobj} from
\emph{PMDK}~\cite{rudoff2014pmdk}, \emph{JEMalloc}~\cite{evans2006jemalloc}, and
\emph{LRMalloc}~\cite{leite2018lrmalloc} (Ralloc without flush and fence). Since
all benchmarks and applications in our experiments used
\code{malloc}/\code{free}, for PMDK's \\\code{malloc-to}/\code{free-from}
interface we had to create a local dummy variable to hold the pointer for easy
integration. For all tests we report the average of three trials.

\subsection{Benchmarks} \label{sec:bench}

\newlength{\microfigwidth}
\begin{figure*}[t!]
    \centering
    \microfigwidth .34\textwidth
    \begin{subfigure}{1.3\microfigwidth}
        \includegraphics[width=1.4\microfigwidth]%
            {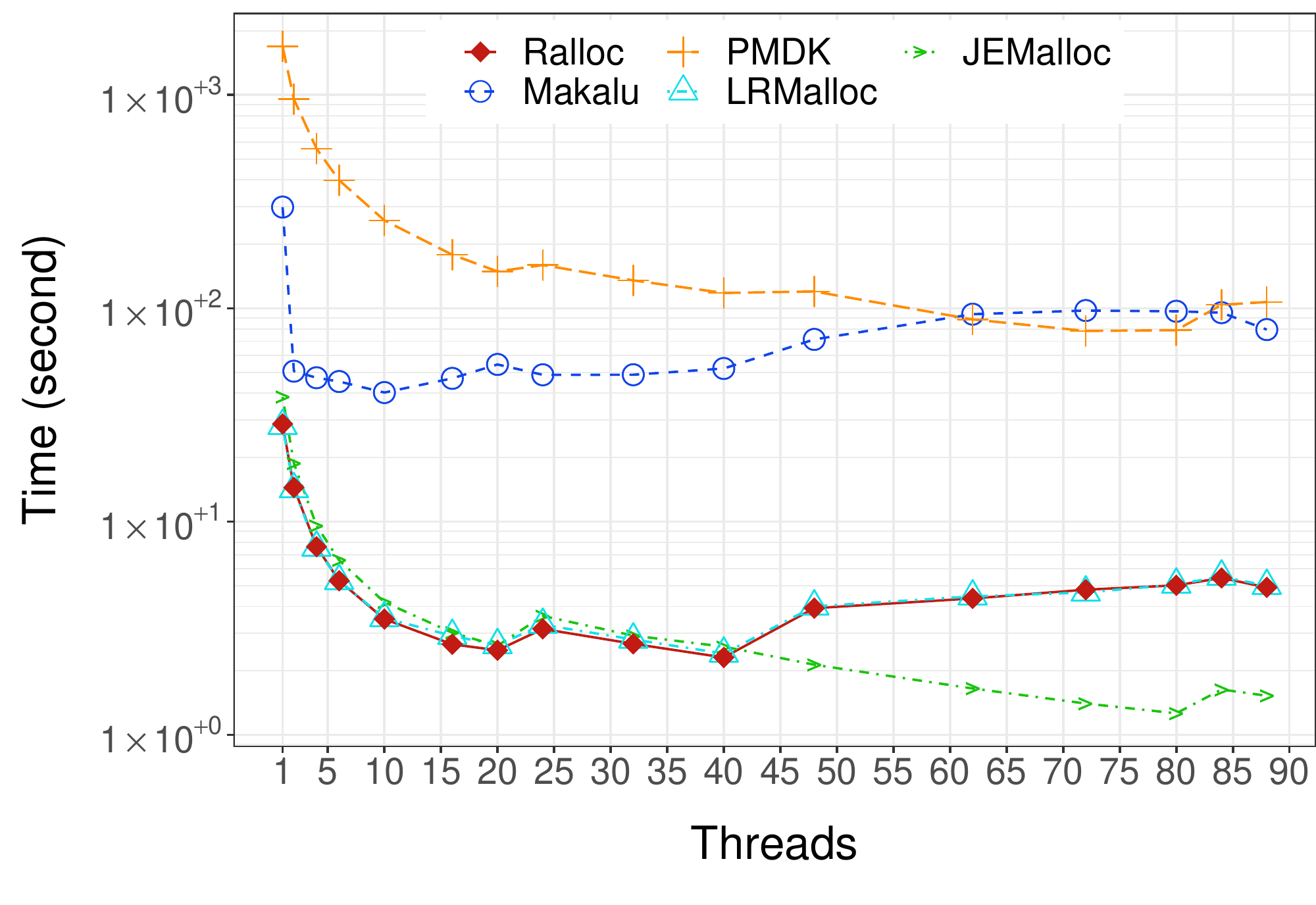}%
        \vspace{-1ex}%
        \caption{Threadtest (lower is better)}
        \label{fig:thd-time}
    \end{subfigure}%
    \hfill
    \begin{subfigure}{1.3\microfigwidth}
        \includegraphics[width=1.4\microfigwidth]%
            {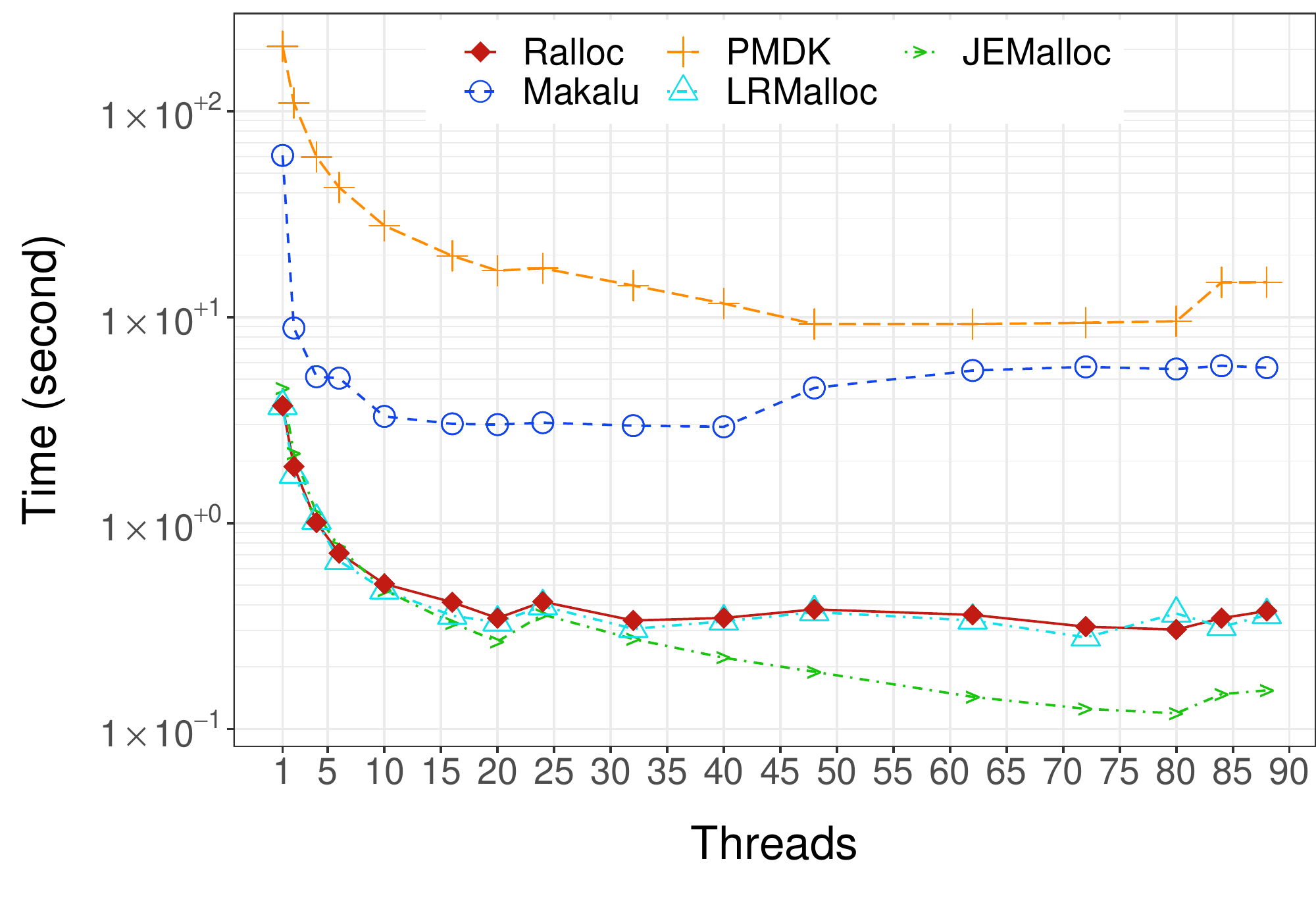}%
        \vspace{-1ex}%
        \caption{Shbench (lower is better)}
        \label{fig:sh-time}
    \end{subfigure}%
    \hfill\strut\\[2ex]
    \smallskip
    \begin{subfigure}{1.3\microfigwidth}
        \includegraphics[width=1.4\microfigwidth]%
            {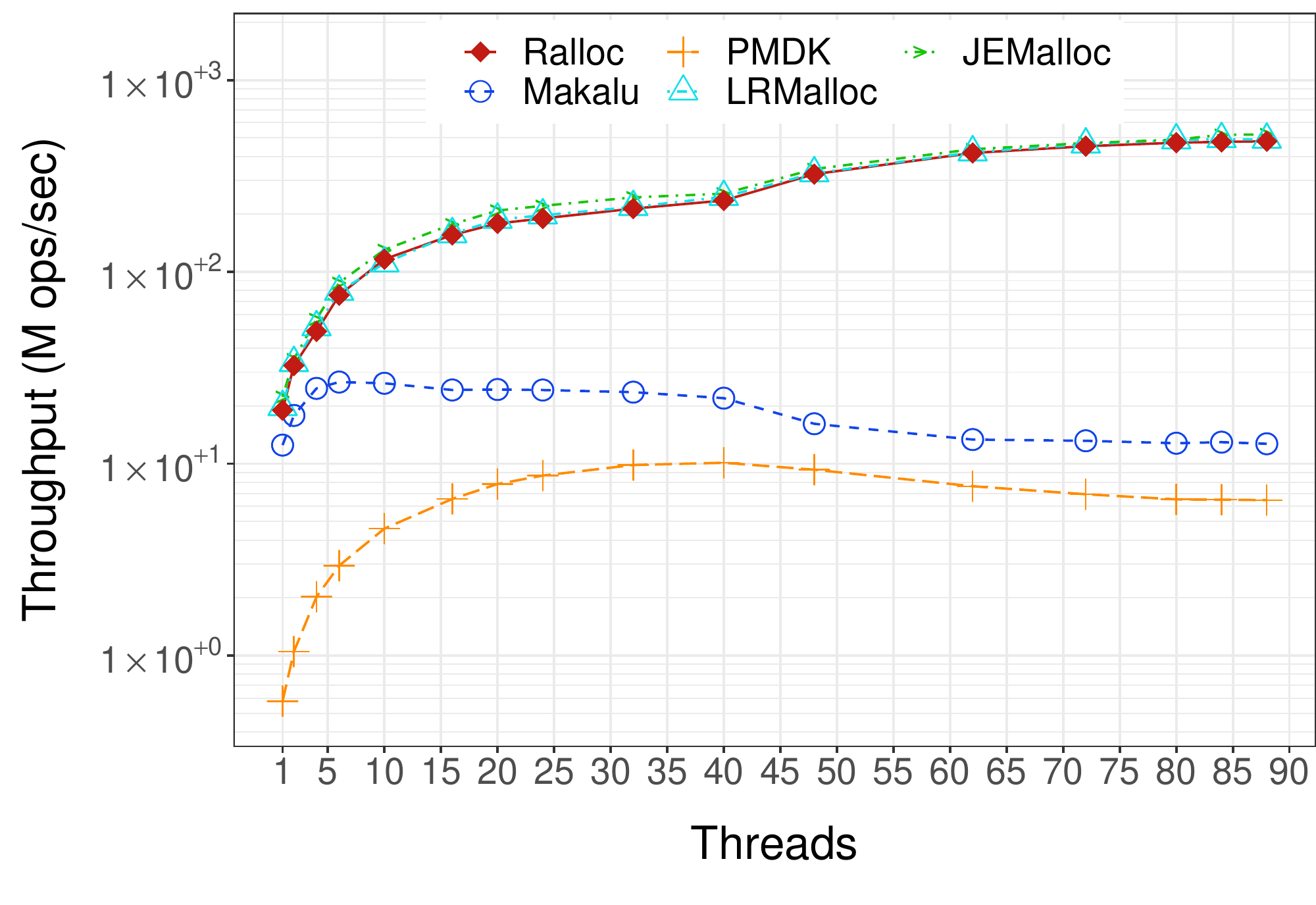}%
        \vspace{-1ex}%
        \caption{Larson (higher is better)}
        \label{fig:larson-through}
    \end{subfigure}%
    \hfill
    \begin{subfigure}{1.3\microfigwidth}
        \includegraphics[width=1.4\microfigwidth]%
            {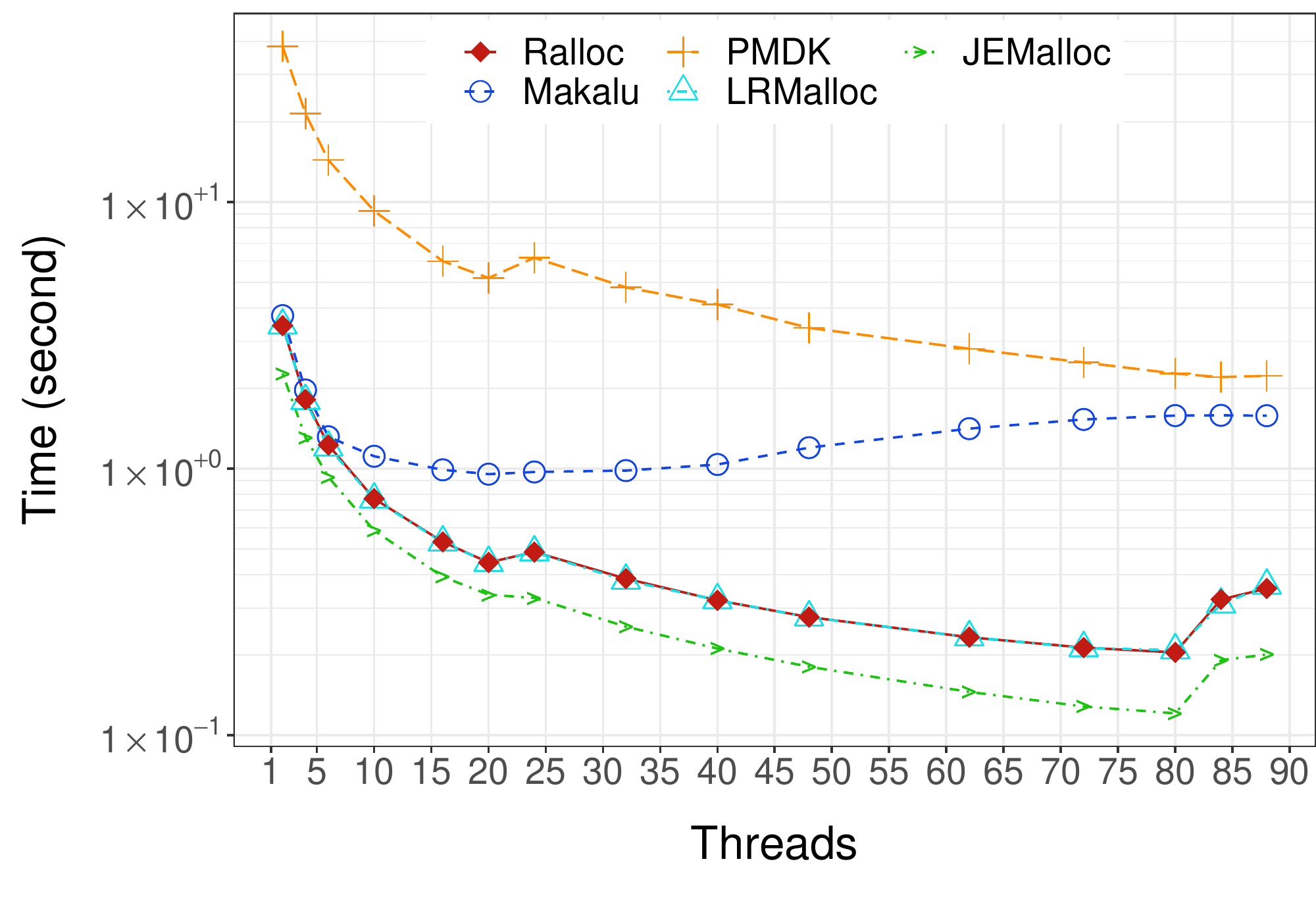}%
        \vspace{-1ex}%
        \caption{Prod-con (lower is better)}
        \label{fig:prod-time}
        \vspace*{-\baselineskip}
    \end{subfigure}%
    \hfill\strut\\[2ex]
    \smallskip
    \begin{subfigure}{1.3\microfigwidth}
        \includegraphics[width=1.4\microfigwidth]%
            {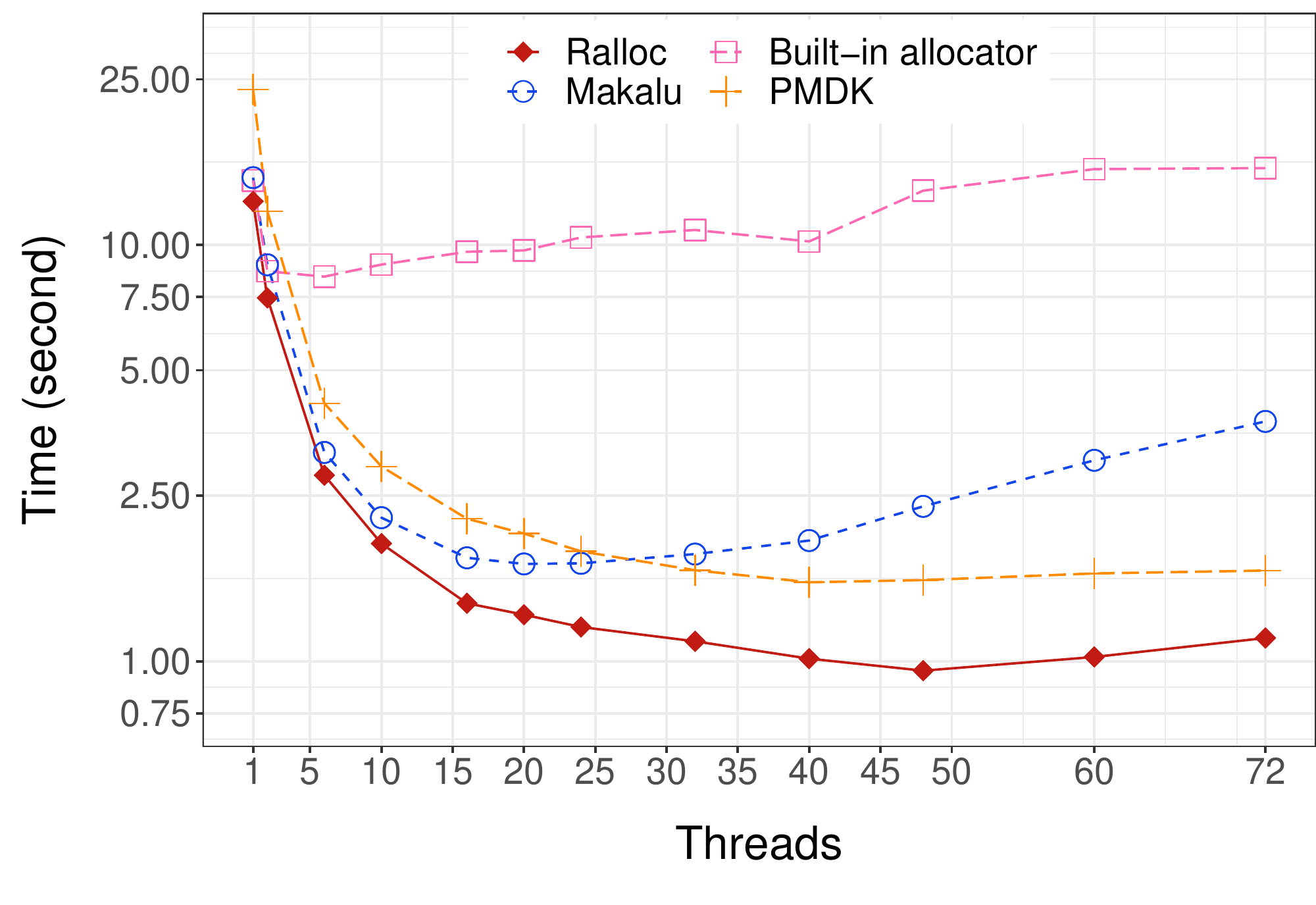}%
        \vspace{-1ex}%
        \caption{Vacation (lower is better)}
        \label{fig:vacation-time}
        \vspace*{-\baselineskip}
    \end{subfigure}%
    \hfill
    \begin{subfigure}{1.3\microfigwidth}
        \includegraphics[width=1.4\microfigwidth]%
            {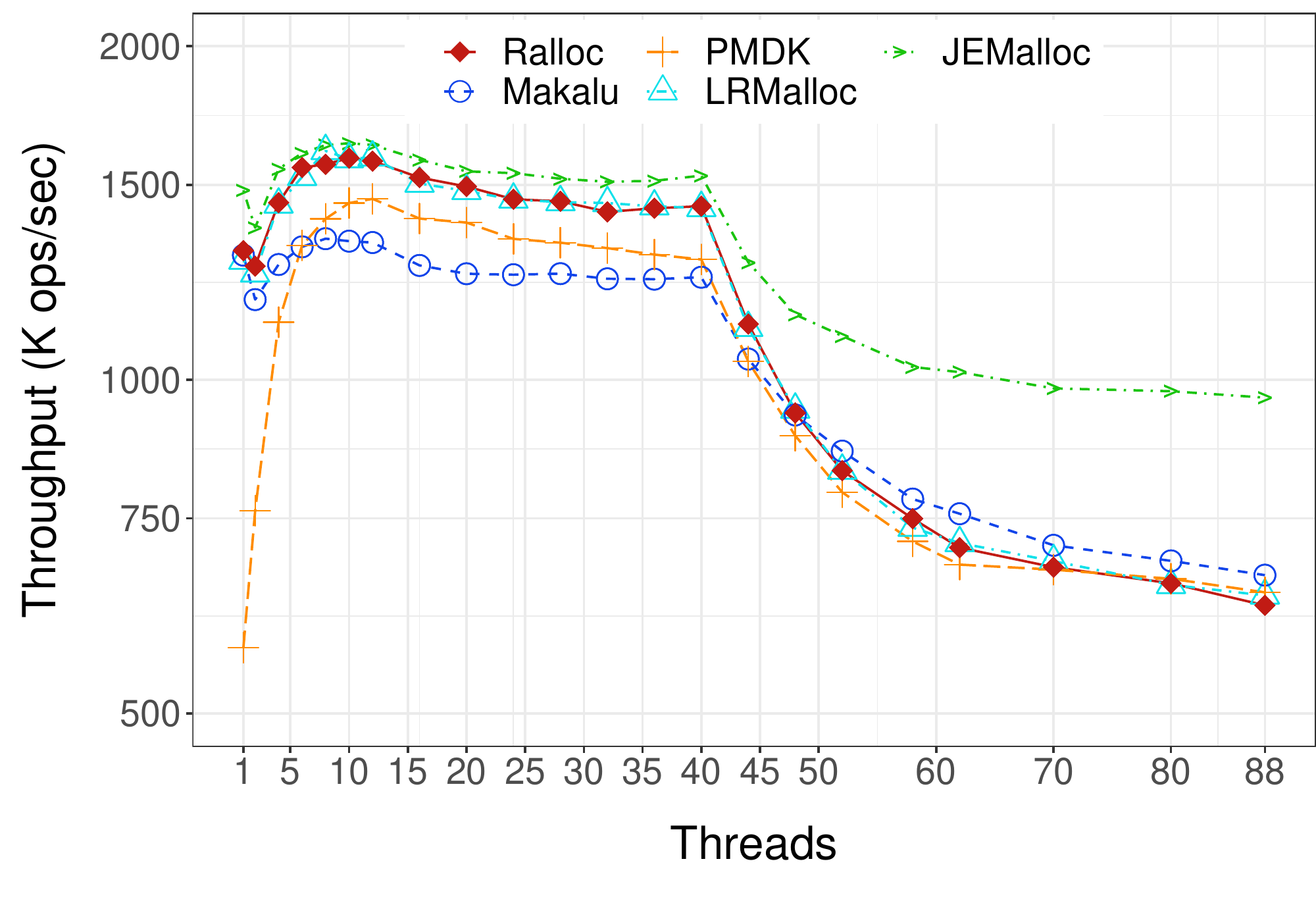}%
        \vspace{-1ex}%
        \caption{Memcached (higher is better)}
        \label{fig:memcached-through}
        \vspace*{-\baselineskip}
    \end{subfigure}%
    \hfill\strut
    \vspace{1ex}
    \caption{Performance (log2 scaled).  Each socket has a
      total of 20 two-way hyperthreaded cores.}
    \label{fig:perf}
\end{figure*}

Our initial evaluation employed
four well known allocator workloads.

\textbf{Threadtest}, introduced with the
\emph{Hoard} allocator~\cite{berger2000hoard},
allocates and deallocates a large number of objects
without any sharing or synchronization between threads.
In every iteration of the test, each thread allocates and deallocates
$10^5$ $64$-byte objects; our experiment comprises $10^4$ iterations.

\textbf{Shbench}~\cite{microquillshbench} is designed as
an allocator stress test.
Threads allocate and deallocate many objects, of sizes that vary from $64$ to
$400$ bytes (the largest of Makalu's ``small'' allocation sizes), with
smaller objects being allocated more frequently. Our experiment
comprises $10^5$ iterations.

\textbf{Larson}~\cite{larson1998larson} simulates a behavior called ``bleeding''
in which some of the objects allocated by one thread are left to be freed by
another thread. This test was configured to spawn $t$ threads that randomly
allocate and deallocate $10^3$ objects in each iteration, ranging in
size from $64$ to $400$ bytes.
After $10^4$ iterations, each thread creates a new thread that
starts with the leftover objects from the previous thread and repeats the same
procedure as before.
Our experiment runs this pattern for $30$ seconds.

\textbf{Prod-con} is a local re-implementation of a test originally
devised for Makalu~\cite{bhandari2016makalu}.
It is meant to assess performance
under a producer-consumer workload. The test spawns $t/2$ pairs of
threads and assigns a lock-free M\&S
queue~\cite{michael1996mlsqueue}
to each pair. One thread in each pair allocates $10^7\!\cdot 2/t$
64\,B objects and pushes pointers to them into the queue; the other thread
concurrently pops pointers from the queue and deallocates the objects.


Performance results appear in Figures~\ref{fig:thd-time}--%
\ref{fig:prod-time}. 
In many cases, curves change shape (generally for the worse) at 20 and
40 threads, as execution moves onto sister hyperthreads and the second
socket, respectively.
The 20-thread inflection point presumably reflects competition for cache and
pipeline resources, the 40-thread inflection point the cost of cross-socket
communication. Overall, Ralloc outperforms and scales better than PMDK and
Makalu on all benchmarks, and is close to JEMalloc for low thread counts.
Makalu, however, usually stops scaling before 20 threads. 

On Threadtest and Shbench, Ralloc performs around 10$\times$ faster than Makalu
and PMDK, presumably because the earlier systems must log and flush multiple
words in synchronized allocator operation, while Ralloc needs no logging at all,
and flushes only occasionally---and then only a single word (the block size
during superblock allocation). 

On Larson, Ralloc performs up to 37$\times$ faster than Makalu. We attribute
this to Makalu's lack of robustness for large numbers of threads. In addition
to the result shown in the figure, we have also tested the allocators on Larson
with a wider range of sizes (64--2048 bytes, the largest ``medium'' allocation
size in Makalu). In this test Makalu stopped scaling after only 2 threads, and
performed up to 100$\times$ slower than Ralloc (1\,M versus 100\,M at 16 threads).
This may suggest that ``medium'' allocations severally compromise Makalu's
scalability.

On Prod-con, Ralloc's performance is close to that of Makalu for low thread counts, but
afterwards scales better. This is because the most of time is spent on
synchronization on the M\&S queues for low thread counts, which covers the
difference in allocation overhead.

\subsection{Application Tests}

We also tested Ralloc on a persistent version of
\emph{Vacation} and a local version of \emph{Memcached}.
Vacation (from the STAMP suite~\cite{chi2008stamp}) is a simulated
online transaction processing system, whose internal ``database'' is implemented
as a set of red-black trees. We obtained the code for this experiment, along
with that of the Mnemosyne~\cite{volos2011mnemosyne} persistent transaction
system, from the University of Wisconsin's WHISPER
suite~\cite{nalli2017whisper}.
Memcached~\cite{libmemcached} is a widely used in-memory key-value store.
We modified it to function as a library rather than a stand-alone
server: instead of sending requests over a socket, the client
application makes direct function calls into the key-value code, much as
it would in a library database like Silo~\cite{tu-sosp13-silo}.  Our
version of memcached can also be shared safely between applications
using the Hodor protected library system~\cite{hedayati2019hodor}; to
focus our attention on allocator performance, we chose not to enable
protection on library calls for the experiments reported here.



The \textbf{Vacation} test employs a total of $16384$ ``relations'' in
its red-black trees. Each transaction comprises $5$ queries, and the
$10^6$ transactions performed by each test target $90\%$ of the relations.
All queries are to create new reservations.
Given that the code had been modified explicitly for persistence, we
tested only the persistent allocators, which exclude LRMalloc and
JEMalloc, but include Mnemosyne's built-in allocator, a persistent
hybrid of Hoard~\cite{berger2000hoard} and DLMalloc~\cite{lea2000dlmalloc}.

The \textbf{Memcached} test runs the Yahoo! Cloud Serving Benchmark
(YCSB) \cite{cooper2010ycsb}, configured to be write-dominant
(workload A~\cite{ycsbworkloads} with 50\% reads and 50\% updates). In
total, 100K operations were executed on 100K records.


Application performance results appear in Figures~\ref{fig:vacation-time}
and~\ref{fig:memcached-through}.  
Results on Vacation resemble those of the allocator benchmarks:
Ralloc scales better than other persistent allocators, and performs
fastest for all sampled thread counts.
Memcached tells a slightly more interesting story:
Performance is relatively flat up to 40 threads, with
Ralloc outperforming both Makalu and PMDK\@.
Performance deteriorates with the cost of cross-socket communication,
however, and Makalu gains a performance edge, outperforming Ralloc by up
to 7\% on 62 threads.
Our best explanation is that Makalu provides slightly better locality
for applications with a large memory footprint.  In particular,
instead of transferring an over-full thread-local free list (cache) back to a
central pool in its entirety, as Ralloc does, Makalu returns only half, allowing
the local thread to continue to allocate from the portion that remains.

On memcached's read-dominant workload
(workload B~\cite{ycsbworkloads} with 95\% reads and 5\% updates---not
shown here), Ralloc continues to outperform Makalu by a small amount at
all thread counts.  The curves are otherwise similar to those in
Figure~\ref{fig:memcached-through}.

\subsection{Recovery}

In a final series of experiments, we measured the cost of Ralloc's
recovery procedure by running an application without calling
\code{close()} at the end, thereby triggering recovery at the start of a
subsequent run.

\begin{figure}
    \centering
    \microfigwidth .34\textwidth
    \begin{subfigure}{1.3\microfigwidth}
        \includegraphics[width=1.4\microfigwidth]%
            {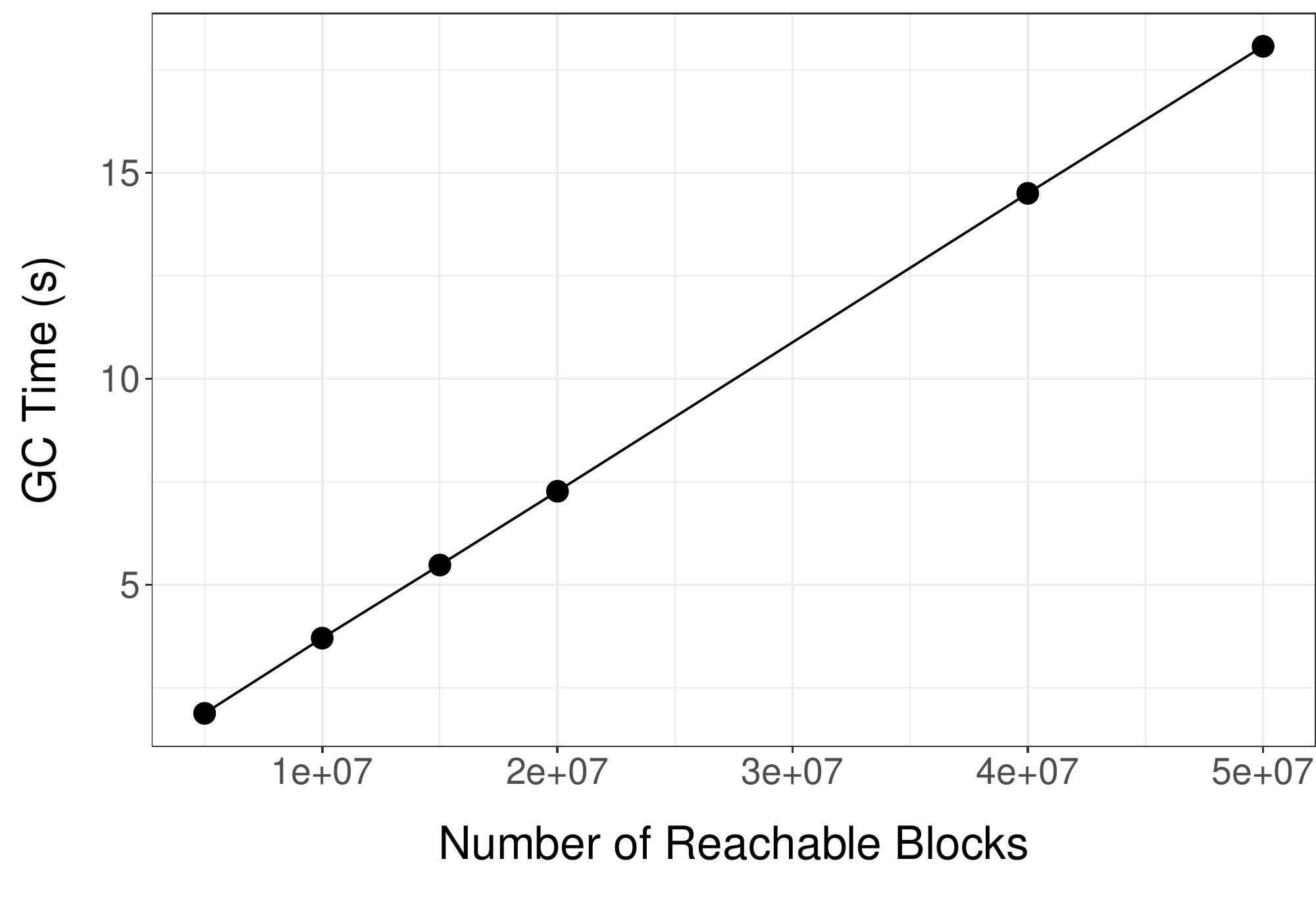}%
        \vspace{-1ex}%
        \caption{Treiber Stack}
        \label{fig:link-rec}
        \vspace*{-\baselineskip}
    \end{subfigure}%
    \hfill
    \begin{subfigure}{1.3\microfigwidth}
        \includegraphics[width=1.4\microfigwidth]%
            {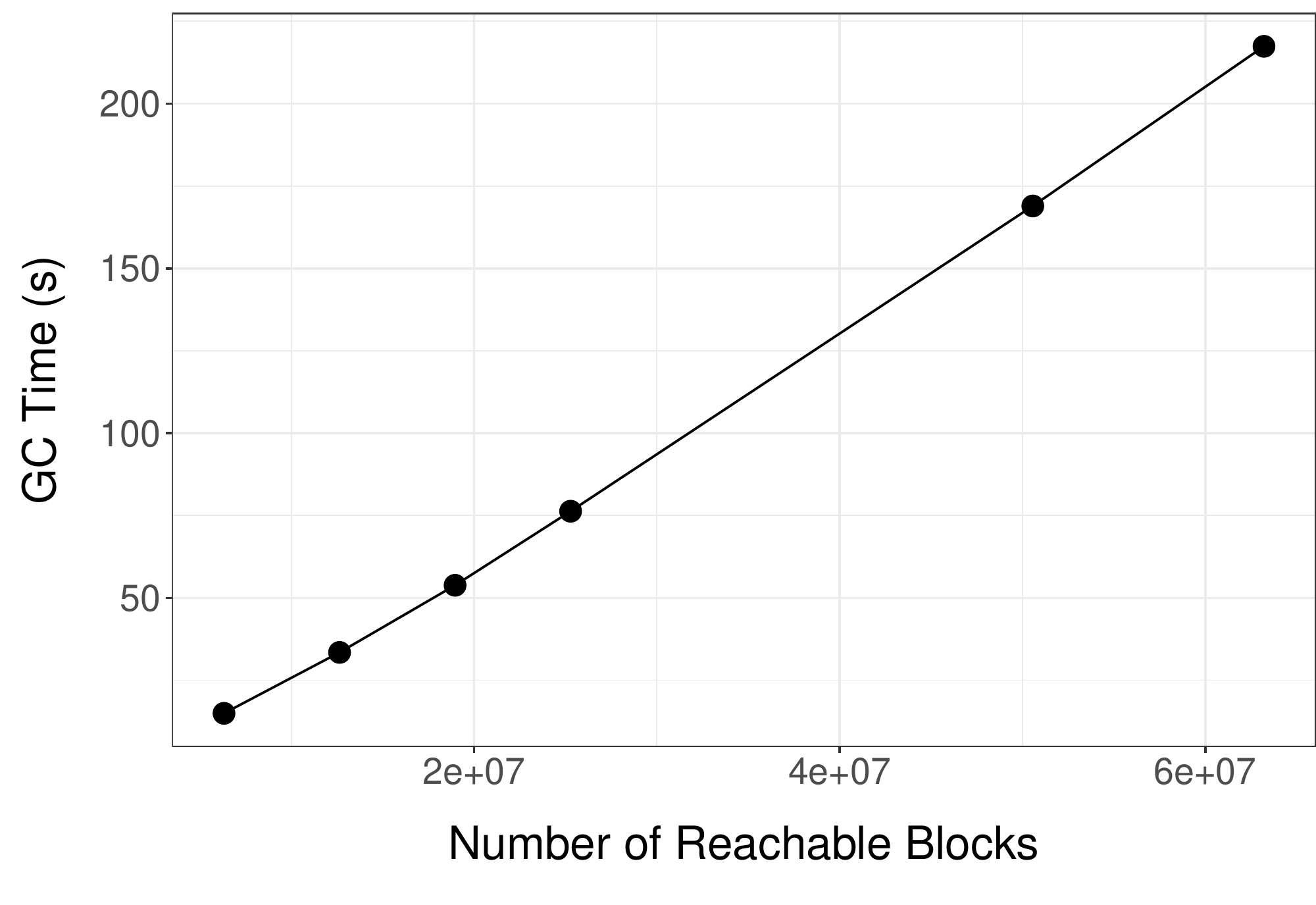}%
        \vspace{-1ex}%
        \caption{Natarajan \& Mittal Tree}
        \label{fig:nat-rec}
        \vspace*{-\baselineskip}
    \end{subfigure}%
    \hfill
    \vspace{1ex}
    \caption{GC Time Consumption.}
    \label{fig:rec}
    \vspace*{-\baselineskip}
\end{figure}

We inserted 
random key-value pairs into a lock-free Treiber stack~\cite{treiber-1986} in one
experiment and, in the other, into the nonblocking binary search tree of
Natarajan and Mittal~\cite{natarajan2014tree}.  We
recorded recovery time for varying numbers of reachable blocks;
results appear in Figure~\ref{fig:rec}.
As expected, recovery time is linear in the number of reachable blocks.
Per-node time is higher in the case of the tree, presumably due to
poorer cache locality.
After recovery, the application was able to restore the structure
correctly in all cases, and to continue performing operations without error.

While we currently run recovery sequentially, it would be
straightforward in the procedure of Section~\ref{sec:rec} to parallelize
Step~\ref{sp:gc} across persistent roots and Steps
\ref{sp:sweep}--\ref{sp:recon} across superblocks; we leave this to
future work.



\section{Conclusions}\label{sec:sum} 

In this paper, we introduced the notion of \emph{recoverability} as a
correctness criterion for persistent memory allocators.
Building on the (transient) LRMalloc nonblocking allocator, we then
presented Ralloc, which we believe to be the first recoverable lock-free
allocator for persistent memory.

As part of Ralloc, we introduced the notion of \emph{filter functions},
which allow a programmer to refine the behavior of conservative garbage
collection without relying on compiler support or per-block
prefixing~\cite{cohen2018recovery}.
We believe that filter functions may be a useful mechanism in other
(e.g., transient) conservative garbage collectors.

Using recovery-time garbage collection,
Ralloc is able to achieve recoverability with almost no
run-time overhead during crash-free execution.
By using offset-based pointers,
Ralloc supports \emph{position-independent
  data} for flexible sharing across executions and among concurrent
processes.

Experimental results show that Ralloc matches or exceeds the performance
of Makalu, the state-of-the-art lock-based persistent allocator, and is
competitive with the well-known JEMalloc transient allocator.
Near-term plans include
the addition of general cross-heap persistent pointers,
integration with persistent libraries~\cite{hedayati2019hodor},
parallelized and optimized recovery,
and detection and (stop-the-world) recovery for independent process
failures.  Longer term, we plan to explore online recovery.



\bibliographystyle{plain}
\bibliography{main}

\begin{thebibliography}{10}

\bibitem{beadle2019qstm}
H.~Alan Beadle, Wentao Cai, Haosen Wen, and Michael~L. Scott.
\newblock Towards efficient nonblocking persistent software transactional
  memory.
\newblock Technical Report TR 1006, Department of Computer Science, Univ. of
  Rochester, April 2019.
\newblock Extended abstract published as a brief announcement at \emph{PPoPP}
  2020.

\bibitem{berger2000hoard}
Emery~D. Berger, Kathryn~S. McKinley, Robert~D. Blumofe, and Paul~R. Wilson.
\newblock Hoard: A scalable memory allocator for multithreaded applications.
\newblock In {\em 9th Intl. Conf. on Architectural Support for Programming
  Languages and Operating Systems (ASPLOS)}, pages 117--128, Cambridge, MA,
  November 2000.

\bibitem{bhandari2016makalu}
Kumud Bhandari, Dhruva~R. Chakrabarti, and Hans-J. Boehm.
\newblock Makalu: Fast recoverable allocation of non-volatile memory.
\newblock In {\em ACM SIGPLAN Intl. Conf. on Object-Oriented Programming,
  Systems, Languages, and Applications (OOPSLA)}, pages 677--694, Amsterdam,
  Netherlands, October 2016.

\bibitem{boehm1988gc}
Hans-Juergen Boehm and Mark Weiser.
\newblock Garbage collection in an uncooperative environment.
\newblock {\em Software: Practice and Experience}, 18(9):807--820, September
  1988.

\bibitem{chi2008stamp}
Chi {Cao Minh}, JaeWoong Chung, Christos Kozyrakis, and Kunle Olukotun.
\newblock {STAMP}: Stanford transactional applications for multi-processing.
\newblock In {\em IEEE Intl. Symp. on Workload Characterization (IISWC)}, pages
  35--46, Seattle, WA, September 2008.

\bibitem{chakrabarti2014atlas}
Dhruva~R. Chakrabarti, Hans-J. Boehm, and Kumud Bhandari.
\newblock Atlas: Leveraging locks for non-volatile memory consistency.
\newblock In {\em ACM Intl. Conf. on Object Oriented Programming Systems
  Languages \& Applications (OOPSLA)}, pages 433--452, Portland, OR, October
  2014.

\bibitem{chen-icpp-2002}
DeQing Chen, Chunqiang Tang, XiangChuan Chen, Sandhya Dwarkadas, and Michael~L.
  Scott.
\newblock Multi-level shared state for distributed systems.
\newblock In {\em Intl. Conf. on Parallel Processing (ICPP)}, pages 131--140,
  Vancouver, BC, Canada, August 2002.

\bibitem{chen2017offholder}
Guoyang Chen, Lei Zhang, Richa Budhiraja, Xipeng Shen, and Youfeng Wu.
\newblock Efficient support of position independence on non-volatile memory.
\newblock In {\em 50th IEEE/ACM Intl. Symp. on Microarchitecture (MICRO)},
  pages 191--203, Cambridge, MA, October 2017.

\bibitem{chinner2015xfs}
Dave Chinner.
\newblock xfs: {DAX} support, March 2015.
\newblock \url{lwn.net/Articles/635514/}.

\bibitem{coburn2011nvheaps}
Joel Coburn, Adrian~M. Caulfield, Ameen Akel, Laura~M. Grupp, Rajesh~K. Gupta,
  Ranjit Jhala, and Steven Swanson.
\newblock {NV-Heaps}: Making persistent objects fast and safe with
  next-generation, non-volatile memories.
\newblock In {\em 16th Intl. Conf. on Architectural Support for Programming
  Languages and Operating Systems (ASPLOS)}, pages 105--118, Newport Beach, CA,
  March 2011.

\bibitem{cohen2018recovery}
Nachshon Cohen, David~T. Aksun, and James~R. Larus.
\newblock Object-oriented recovery for non-volatile memory.
\newblock {\em Proceedings of the ACM on Programming Languages},
  2(OOPSLA):153:1--153:22, October 2018.

\bibitem{ycsbworkloads}
Brian Cooper.
\newblock {YCSB} core workloads, 2010.
\newblock \url{github.com/brianfrankcooper/YCSB/wiki/Core-Workloads}.

\bibitem{cooper2010ycsb}
Brian~F. Cooper, Adam Silberstein, Erwin Tam, Raghu Ramakrishnan, and Russell
  Sears.
\newblock Benchmarking cloud serving systems with {YCSB}.
\newblock In {\em 1st ACM Symp. on Cloud Computing (SoCC)}, pages 143--154,
  Indianapolis, IN, June 2010.

\bibitem{correia2018romulus}
Andreia Correia, Pascal Felber, and Pedro Ramalhete.
\newblock Romulus: Efficient algorithms for persistent transactional memory.
\newblock In {\em 30th on Symp. on Parallelism in Algorithms and Architectures
  (SPAA)}, pages 271--282, Vienna, Austria, July 2018.

\bibitem{hpe2017atlasmakalu}
Hewlett~Packard Enterprise.
\newblock makalu\_alloc, May 2017.
\newblock \url{github.com/HewlettPackard/Atlas/tree/makalu/makalu_alloc}.

\bibitem{evans2006jemalloc}
Jason Evans.
\newblock A scalable concurrent malloc (3) implementation for {FreeBSD}.
\newblock In {\em BSDCan Conf.}, Ottawa, Ontario, Canada, May 2006.

\bibitem{fraser-thesis-2004}
Keir Fraser.
\newblock {\em Practical Lock-Freedom}.
\newblock PhD thesis, King's College, Univ. of Cambridge, September 2003.
\newblock Published as Univ. of Cambridge Computer Lab technical report \#579,
  Feb. 2004. \url{www.cl.cam.ac.uk/techreports/UCAM-CL-TR-579.pdf}.

\bibitem{friedman2018persistent}
Michal Friedman, Maurice Herlihy, Virendra Marathe, and Erez Petrank.
\newblock A persistent lock-free queue for non-volatile memory.
\newblock In {\em 23rd ACM SIGPLAN Symp. on Principles and Practice of Parallel
  Programming (PPoPP)}, pages 28--40, Vienna, Austria, February 2018.

\bibitem{gidenstam2010nbmaloc}
Anders Gidenstam, Marina Papatriantafilou, and Philippas Tsigas.
\newblock {NBmalloc}: Allocating memory in a lock-free manner.
\newblock {\em Algorithmica}, 58(2):304--338, Oct 2010.

\bibitem{hedayati2019hodor}
Mohammad Hedayati, Spyridoula Gravani, Ethan Johnson, John Criswell, Michael~L.
  Scott, Kai Shen, and Mike Marty.
\newblock Hodor: Intra-process isolation for high-throughput data plane
  libraries.
\newblock In {\em {USENIX} Annual Technical Conf. (ATC)}, pages 489--504,
  Renton, WA, July 2019.

\bibitem{hudson-ismm-2006}
Richard~L. Hudson, Bratin Saha, Ali-Reza Adl-Tabatabai, and Benjamin Hertzberg.
\newblock {McRT-Malloc}---a scalable transactional memory allocator.
\newblock In {\em Intl. Symp. on Memory Management (ISMM)}, pages 74--83,
  Ottawa, ON, Canada, June 2006.

\bibitem{microquillshbench}
MicroQuill Inc.
\newblock shbench, 2014.
\newblock \url{www.microquill.com/smartheap/shbench/}.

\bibitem{intel-manual}
{Intel Corp.}
\newblock {\em {Intel 64 and IA-32 Architectures Software Developer's Manual}},
  May 2019.
\newblock 325462-070US.

\bibitem{izraelevitz2016justdo}
Joseph Izraelevitz, Terence Kelly, and Aasheesh Kolli.
\newblock Failure-atomic persistent memory updates via {JUSTDO} logging.
\newblock In {\em 21st Intl. Conf. on Architectural Support for Programming
  Languages and Operating Systems (ASPLOS)}, pages 427--442, Atlanta, GA, April
  2016.

\bibitem{izraelevitz2016linearizability}
Joseph Izraelevitz, Hammurabi Mendes, and Michael~L. Scott.
\newblock Linearizability of persistent memory objects under a
  full-system-crash failure model.
\newblock In {\em Intl. Symp. on Distributed Computing (DISC)}, pages 313--327,
  Paris, France, September 2016.

\bibitem{izraelevitz2019performance}
Joseph Izraelevitz, Jian Yang, Lu~Zhang, Juno Kim, Xiao Liu, Amirsaman
  Memaripour, Yun~Joon Soh, Zixuan Wang, Yi~Xu, Subramanya~R. Dulloor, Jishen
  Zhao, and Steven Swanson.
\newblock Basic performance measurements of the {Intel Optane} {DC} persistent
  memory module.
\newblock {\em CoRR}, abs/1903.05714, 2019.

\bibitem{larson1998larson}
{Per-\r{A}ke} Larson and Murali Krishnan.
\newblock Memory allocation for long-running server applications.
\newblock In {\em 1st Intl. Symp. on Memory Management (ISMM)}, pages 176--185,
  Vancouver, BC, Canada, October 1998.

\bibitem{lea2000dlmalloc}
Doug Lea.
\newblock A memory allocator, April 2000.
\newblock \url{gee.cs.oswego.edu/dl/html/malloc.html}.

\bibitem{leite2018lrmalloc}
Ricardo Leite and Ricardo Rocha.
\newblock {LRMalloc}: A modern and competitive lock-free dynamic memory
  allocator.
\newblock In {\em 13th Intl. Meeting on High Performance Computing for
  Computational Science (VECPAR)}, pages 230--243, São Pedro, São Paulo,
  Brazil, September 2018.

\bibitem{libmemcached}
libMemcached.org.
\newblock libmemcached, 2011.
\newblock \url{libmemcached.org/libMemcached.html}.

\bibitem{liu2018ido}
Qingrui Liu, Joseph Izraelevitz, Se~Kwon Lee, Michael~L. Scott, Sam~H. Noh, and
  Changhee Jung.
\newblock {iDO}: Compiler-directed failure atomicity for nonvolatile memory.
\newblock In {\em 51st IEEE/ACM Intl. Symp. on Microarchitecture (MICRO)},
  pages 258--270, Fukuoka, Japan, October 2018.

\bibitem{michael2004hazard}
Maged~M. Michael.
\newblock Hazard pointers: Safe memory reclamation for lock-free objects.
\newblock {\em IEEE Transactions on Parallel and Distributed Systems},
  15(6):491--504, June 2004.

\bibitem{michael2004scalable}
Maged~M. Michael.
\newblock Scalable lock-free dynamic memory allocation.
\newblock In {\em ACM SIGPLAN Conf. on Programming Language Design and
  Implementation (PLDI)}, pages 35--46, Washington DC, June 2004.

\bibitem{michael1996mlsqueue}
Maged~M. Michael and Michael~L. Scott.
\newblock Simple, fast, and practical non-blocking and blocking concurrent
  queue algorithms.
\newblock In {\em 15th ACM Symp. on Principles of Distributed Computing
  (PODC)}, pages 267--275, Philadelphia, PA, May 1996.

\bibitem{nalli2017whisper}
Sanketh Nalli, Swapnil Haria, Mark~D. Hill, Michael~M. Swift, Haris Volos, and
  Kimberly Keeton.
\newblock An analysis of persistent memory use with {WHISPER}.
\newblock In {\em 22nd Intl. Conf. on Architectural Support for Programming
  Languages and Operating Systems (ASPLOS)}, pages 135--148, Xi'an, China,
  2017.

\bibitem{natarajan2014tree}
Aravind Natarajan and Neeraj Mittal.
\newblock Fast concurrent lock-free binary search trees.
\newblock In {\em 19th ACM SIGPLAN Symp. on Principles and Practice of Parallel
  Programming (PPoPP)}, pages 317--328, Orlando, FL, February 2014.

\bibitem{nawab2017dali}
Faisal Nawab, Joseph Izraelevitz, Terence Kelly, Charles~B. {Morrey III},
  Dhruva~R. Chakrabarti, and Michael~L. Scott.
\newblock Dal{\'i}: A periodically persistent hash map.
\newblock In {\em Intl. Symp. on Distributed Computing (DISC)}, volume~91,
  pages 37:1--37:16, Vienna, Austria, October 2017.

\bibitem{oukid2017pallocator}
Ismail Oukid, Daniel Booss, Adrien Lespinasse, Wolfgang Lehner, Thomas
  Willhalm, and Gr{\'e}goire Gomes.
\newblock Memory management techniques for large-scale persistent-main-memory
  systems.
\newblock {\em Proceedings of the VLDB Endowment}, 10(11):1166--1177, August
  2017.

\bibitem{ramalhete2019onefile}
Pedro Ramalhete, Andreia Correia, Pascal Felber, and Nachshon Cohen.
\newblock {OneFile}: A wait-free persistent transactional memory.
\newblock In {\em 49th IEEE/IFIP Intl. Conf. on Dependable Systems and Networks
  (DSN)}, pages 151--163, Portland, OR, June 2019.

\bibitem{rudoff2017persistent}
Andy Rudoff.
\newblock Persistent memory programming.
\newblock {\em Login: The Usenix Magazine}, 42:34--40, 2017.

\bibitem{rudoff2014pmdk}
Andy Rudoff and Marcin Slusarz.
\newblock Persistent memory development kit, September 2014.
\newblock \url{pmem.io/pmdk/}.

\bibitem{SMS}
Michael~L. Scott.
\newblock {\em Shared-Memory Synchronization}.
\newblock Morgan \& Claypool, 2013.

\bibitem{sangmin2011sfmalloc}
Sangmin Seo, Junghyun Kim, and Jaejin Lee.
\newblock {SFMalloc}: A lock-free and mostly synchronization-free dynamic
  memory allocator for manycores.
\newblock In {\em 2011 Intl. Conf. on Parallel Architectures and Compilation
  Techniques (PACT)}, pages 253--263, Galveston, TX, October 2011.

\bibitem{shacham2004random}
Hovav Shacham, Matthew Page, Ben Pfaff, Eu-Jin Goh, Nagendra Modadugu, and Dan
  Boneh.
\newblock On the effectiveness of address-space randomization.
\newblock In {\em 11th ACM Conf. on Computer and Communications Security
  (CCS)}, pages 298--307, Washington DC, October 2004.

\bibitem{pthread-mutex}
{The Open Group}.
\newblock \code{pthread\_mutex\_lock}.
\newblock {IEEE Std 1003.1-2017}, 2018.

\bibitem{treiber-1986}
R.~Kent Treiber.
\newblock Systems programming: Coping with parallelism.
\newblock Technical Report RJ 5118, IBM Almaden Research Center, April 1986.

\bibitem{tu-sosp13-silo}
Stephen Tu, Wenting Zheng, Eddie Kohler, Barbara Liskov, and Samuel Madden.
\newblock {Speedy Transactions in Multicore In-memory Databases}.
\newblock In {\em 24th ACM Symp. on Operating Systems Principles (SOSP)}, pages
  18--32, Farmington, PA, November 2013.

\bibitem{erim}
Anjo Vahldiek-Oberwagner, Eslam Elnikety, Nuno~O. Duarte, Michael Sammler,
  Peter Dru\-schel, and Deepak Garg.
\newblock {ERIM}: Secure, efficient in-process isolation with protection keys
  {(MPK)}.
\newblock In {\em 28th USENIX Security Symp. (SEC)}, pages 1221--1238, Santa
  Clara, CA, August 2019.

\bibitem{vasudevan-pact-2010}
Nalini Vasudevan, Kedar~S. Namjoshi, and Stephen~A. Edwards.
\newblock Simple and fast biased locks.
\newblock In {\em 19th Intl. Conf. on Parallel Architectures and Compilation
  Techniques (PACT)}, pages 65--74, Vienna, Austria, September 2010.

\bibitem{volos2011mnemosyne}
Haris Volos, Andres~Jaan Tack, and Michael~M. Swift.
\newblock Mnemosyne: Lightweight persistent memory.
\newblock In {\em 16th Intl. Conf. on Architectural Support for Programming
  Languages and Operating Systems (ASPLOS)}, pages 91--104, Newport Beach, CA,
  March 2011.

\bibitem{wen2018ibr}
Haosen Wen, Joseph Izraelevitz, Wentao Cai, H.~Alan Beadle, and Michael~L.
  Scott.
\newblock Interval-based memory reclamation.
\newblock In {\em 23th ACM SIGPLAN Symp. on Principles and Practice of Parallel
  Programming (PPoPP)}, pages 1--13, Vienna, Austria, 2018.

\bibitem{wilcox2017ext4}
Matthew Wilcox.
\newblock Add support for {NV-DIMMs} to {Ext4}, December 2017.
\newblock \url{kernelnewbies.org/Ext4}.

\bibitem{williams2019persist}
Dan Williams.
\newblock Persistent memory, August 2019.
\newblock \url{nvdimm.wiki.kernel.org/}.

\bibitem{xu2016nova}
Jian Xu and Steven Swanson.
\newblock {NOVA}: A log-structured file system for hybrid volatile/non-volatile
  main memories.
\newblock In {\em 14th {USENIX} Conf. on File and Storage Technologies (FAST)},
  pages 323--338, Santa Clara, CA, February 2016.

\bibitem{yang2019orion}
Jian Yang, Joseph Izraelevitz, and Steven Swanson.
\newblock Orion: A distributed file system for non-volatile main memory and
  {RDMA}-capable networks.
\newblock In {\em 17th {USENIX} Conf. on File and Storage Technologies (FAST)},
  pages 221--234, Boston, MA, February 2019.

\end{thebibliography}

\end{document}